\RequirePackage{snapshot}
\documentclass[11pt,twoside,british]{article}
\usepackage[T1]{fontenc}
\usepackage[latin9]{inputenc}
\usepackage[a4paper]{geometry}
\usepackage{fancyhdr}
\usepackage{color}
\usepackage{babel}
\usepackage{verbatim}
\usepackage{refstyle}
\usepackage{booktabs}
\usepackage{mathrsfs}
\usepackage{mathtools}
\usepackage{url}
\usepackage{enumitem}
\usepackage{amsmath}
\usepackage{amsthm}
\usepackage{amssymb}
\usepackage{graphicx}
\usepackage{setspace}
\usepackage[authoryear,longnamesfirst]{natbib}
\usepackage{xargs}[2008/03/08]
\usepackage[pdfusetitle,
 bookmarks=true,bookmarksnumbered=false,bookmarksopen=false,
 breaklinks=false,pdfborder={0 0 0},pdfborderstyle={},backref=false,colorlinks=false]
 {hyperref}

\makeatletter

\AtBeginDocument{\providecommand\exaref[1]{\ref{exa:#1}}}
\AtBeginDocument{\providecommand\subsecref[1]{\ref{subsec:#1}}}
\AtBeginDocument{\providecommand\lemref[1]{\ref{lem:#1}}}
\AtBeginDocument{\providecommand\propref[1]{\ref{prop:#1}}}
\AtBeginDocument{\providecommand\secref[1]{\ref{sec:#1}}}
\AtBeginDocument{\providecommand\assref[1]{\ref{ass:#1}}}
\AtBeginDocument{\providecommand\eqref[1]{\ref{eq:#1}}}
\AtBeginDocument{\providecommand\enuref[1]{\ref{enu:#1}}}
\AtBeginDocument{\providecommand\figref[1]{\ref{fig:#1}}}
\AtBeginDocument{\providecommand\appref[1]{\ref{app:#1}}}
\AtBeginDocument{\providecommand\thmref[1]{\ref{thm:#1}}}
\AtBeginDocument{\providecommand\corref[1]{\ref{cor:#1}}}
\AtBeginDocument{\providecommand\tblref[1]{\ref{tbl:#1}}}
\AtBeginDocument{\providecommand\remref[1]{\ref{rem:#1}}}
\providecommand{\tabularnewline}{\\}
\RS@ifundefined{subsecref}
  {\newref{subsec}{name = \RSsectxt}}
  {}
\RS@ifundefined{thmref}
  {\def\RSthmtxt{theorem~}\newref{thm}{name = \RSthmtxt}}
  {}
\RS@ifundefined{lemref}
  {\def\RSlemtxt{lemma~}\newref{lem}{name = \RSlemtxt}}
  {}

\newlength{\lyxlabelwidth}      % auxiliary length 
\numberwithin{equation}{section}
\numberwithin{figure}{section}
\numberwithin{table}{section}
	\newenvironment{elabeling}[2][]%
	{\settowidth{\lyxlabelwidth}{#2}
		\begin{description}[font=\normalfont,style=sameline,
			leftmargin=\lyxlabelwidth,#1]}
	{\end{description}}
\theoremstyle{remark}
\newtheorem*{notation*}{\protect\notationname}
\theoremstyle{definition}
\newtheorem{example}{\protect\examplename}[section]
\theoremstyle{plain}
\newtheorem{assumption}{\protect\assumptionname}
\theoremstyle{remark}
\newtheorem{rem}{\protect\remarkname}[section]
\theoremstyle{plain}
\newtheorem{prop}{\protect\propositionname}[section]
\theoremstyle{definition}
\newtheorem{defn}{\protect\definitionname}[section]
\theoremstyle{plain}
\newtheorem{thm}{\protect\theoremname}[section]
\theoremstyle{plain}
\newtheorem{cor}{\protect\corollaryname}[section]
\theoremstyle{plain}
\newtheorem{lem}{\protect\lemmaname}[section]

\setlist[enumerate,1]{label=\upshape{(\roman*)}, ref=(\roman*)}
\setlist[enumerate,2]{label=\upshape{(\alph*)}, ref=(\alph*)}
\setlist[enumerate,3]{label=\upshape{\roman*.}, ref=\roman*}

\makeatletter
  \@ifpackageloaded{refstyle}{}{\usepackage{refstyle}}
\makeatother

\newref{thm}{name = Theorem~}
\newref{prop}{name = Proposition~}
\newref{lem}{name = Lemma~}
\newref{cor}{name = Corollary~}
\newref{ass}{name = }
\newref{exa}{name = Example~}
\newref{rem}{name = Remark~}
\newref{def}{name = Definition~}

\newref{eq}{refcmd = (\ref{#1})}

\newref{sec}{name = Section~}
\newref{sub}{name = Section~}
\newref{subsec}{name = Section~}
\newref{app}{name = Appendix~}

\newref{fig}{name = Figure~}
\newref{tbl}{name = Table~}

\makeatletter
  \@ifpackageloaded{mathrsfs}{}{\usepackage{mathrsfs}}
\makeatother

\date{}

\usepackage{nextpage}

\allowdisplaybreaks[1]

\usepackage{scalefnt}
\usepackage[group-separator={,}]{siunitx}
\newcommand\smaller[2][0.85]{{\scalefont{#1}#2}}
\newcommand{\ass}[1]{{\upshape{\smaller[0.76]{#1}}}}

\newcommand{\assumpname}[1]{%
  \renewcommand{\theassumption}{\ass{#1}}%
}

\makeatletter
\newsavebox{\@brx}
\newcommand{\dbllangle}[1][]{\savebox{\@brx}{\(\m@th{#1\langle}\)}%
  \mathopen{\copy\@brx\kern-0.5\wd\@brx\usebox{\@brx}}}
\newcommand{\dblrangle}[1][]{\savebox{\@brx}{\(\m@th{#1\rangle}\)}%
  \mathclose{\copy\@brx\kern-0.5\wd\@brx\usebox{\@brx}}}
\makeatother

\usepackage{changepage}

\theoremstyle{definition}
\newtheorem{examplex}{Example}
\renewenvironment{example}
  {\pushQED{\qed}\examplex}
  {\popQED\endexamplex}

\numberwithin{examplex}{section}

\newcommand{\exname}[1]{%
  \renewcommand{\theexamplex}{{#1}}%
}

\newcounter{subremark}[rem]
\renewcommand{\thesubremark}{(\roman{subremark})}

\newcommand{\subremark}{%
  \refstepcounter{subremark}%
  \thesubremark{}.%
}

\newcounter{savedexnumber}

\newcommand{\saveexamplex}{%
  \setcounter{savedexnumber}{\value{examplex}}
}

\newcommand{\restoreexamplex}{%
  \setcounter{examplex}{\value{savedexnumber}}
  \numberwithin{examplex}{section}
}

\usepackage{mathdots}

\usepackage[usestackEOL]{stackengine}
\setstackgap{S}{5pt}
\newcommand{\authaffil}[2]{\Shortunderstack{#1\\\small{#2}}}

\usepackage{needspace}

\newcommand{\R}{\textsf{R}}

\usepackage{bibunits}
\defaultbibliography{../cksvar,additional-examples-cx}
\defaultbibliographystyle{ecta}

\makeatletter  
\renewcommand{\@bibunitname}{\jobname.\the\@bibunitauxcnt}
\makeatother

\usepackage[rightbars]{changebar}

\newif\ifmarkchanges
\markchangestrue
\cbcolor{red}
\makeatother

\providecommand{\assumptionname}{Assumption}
\providecommand{\corollaryname}{Corollary}
\providecommand{\definitionname}{Definition}
\providecommand{\examplename}{Example}
\providecommand{\lemmaname}{Lemma}
\providecommand{\notationname}{Notation}
\providecommand{\propositionname}{Proposition}
\providecommand{\remarkname}{Remark}
\providecommand{\theoremname}{Theorem}

\begin{document}
% Macros, Version 4
% Updated: 2020-11-19

% FORMATTING

\global\long\def\uwrite#1#2{\underset{#2}{\underbrace{#1}} }%

\global\long\def\blw#1{\ensuremath{\underline{#1}}}%

\global\long\def\abv#1{\ensuremath{\overline{#1}}}%

\global\long\def\vect#1{\mathbf{#1}}%

% SETS AND SEQUENCES

\global\long\def\smlseq#1{\{#1\} }%

\global\long\def\seq#1{\left\{  #1\right\}  }%

\global\long\def\smlsetof#1#2{\{#1\mid#2\} }%

\global\long\def\setof#1#2{\left\{  #1\mid#2\right\}  }%

% LIMITS

\global\long\def\goesto{\ensuremath{\rightarrow}}%

\global\long\def\ngoesto{\ensuremath{\nrightarrow}}%

\global\long\def\uto{\ensuremath{\uparrow}}%

\global\long\def\dto{\ensuremath{\downarrow}}%

\global\long\def\uuto{\ensuremath{\upuparrows}}%

\global\long\def\ddto{\ensuremath{\downdownarrows}}%

\global\long\def\ulrto{\ensuremath{\nearrow}}%

\global\long\def\dlrto{\ensuremath{\searrow}}%

% FUNCTIONS AND FUNCTION SPACES

\global\long\def\setmap{\ensuremath{\rightarrow}}%

\global\long\def\elmap{\ensuremath{\mapsto}}%

\global\long\def\compose{\ensuremath{\circ}}%

\global\long\def\cont{C}%

\global\long\def\cadlag{D}%

\global\long\def\Ellp#1{\ensuremath{\mathcal{L}^{#1}}}%

% SETS OF NUMBERS

\global\long\def\naturals{\ensuremath{\mathbb{N}}}%

\global\long\def\reals{\mathbb{R}}%

\global\long\def\complex{\mathbb{C}}%

\global\long\def\rationals{\mathbb{Q}}%

\global\long\def\integers{\mathbb{Z}}%

% NORMS, MODULI, INNER PRODUCTS

\global\long\def\abs#1{\ensuremath{\left|#1\right|}}%

\global\long\def\smlabs#1{\ensuremath{\lvert#1\rvert}}%
 
\global\long\def\bigabs#1{\ensuremath{\bigl|#1\bigr|}}%
 
\global\long\def\Bigabs#1{\ensuremath{\Bigl|#1\Bigr|}}%
 
\global\long\def\biggabs#1{\ensuremath{\biggl|#1\biggr|}}%

\global\long\def\norm#1{\ensuremath{\left\Vert #1\right\Vert }}%

\global\long\def\smlnorm#1{\ensuremath{\lVert#1\rVert}}%
 
\global\long\def\bignorm#1{\ensuremath{\bigl\|#1\bigr\|}}%
 
\global\long\def\Bignorm#1{\ensuremath{\Bigl\|#1\Bigr\|}}%
 
\global\long\def\biggnorm#1{\ensuremath{\biggl\|#1\biggr\|}}%

\global\long\def\floor#1{\left\lfloor #1\right\rfloor }%
\global\long\def\smlfloor#1{\lfloor#1\rfloor}%

\global\long\def\ceil#1{\left\lceil #1\right\rceil }%
\global\long\def\smlceil#1{\lceil#1\rceil}%

% SET OPERATIONS

\global\long\def\Union{\ensuremath{\bigcup}}%

\global\long\def\Intsect{\ensuremath{\bigcap}}%

\global\long\def\union{\ensuremath{\cup}}%

\global\long\def\intsect{\ensuremath{\cap}}%

\global\long\def\pset{\ensuremath{\mathcal{P}}}%

\global\long\def\clsr#1{\ensuremath{\overline{#1}}}%

\global\long\def\symd{\ensuremath{\Delta}}%

\global\long\def\intr{\operatorname{int}}%

\global\long\def\cprod{\otimes}%

\global\long\def\Cprod{\bigotimes}%

% HILBERT SPACES

\global\long\def\smlinprd#1#2{\ensuremath{\langle#1,#2\rangle}}%

\global\long\def\inprd#1#2{\ensuremath{\left\langle #1,#2\right\rangle }}%

\global\long\def\orthog{\ensuremath{\perp}}%

\global\long\def\dirsum{\ensuremath{\oplus}}%

% LINEAR ALGEBRA

\global\long\def\spn{\operatorname{sp}}%

\global\long\def\rank{\operatorname{rk}}%

\global\long\def\proj{\operatorname{proj}}%

\global\long\def\tr{\operatorname{tr}}%

\global\long\def\vek{\operatorname{vec}}%

\global\long\def\diag{\operatorname{diag}}%

\global\long\def\col{\operatorname{col}}%

% PROBABILITY SPACES AND SIGMA-FIELDS

\global\long\def\smpl{\ensuremath{\Omega}}%

\global\long\def\elsmp{\ensuremath{\omega}}%

\global\long\def\sigf#1{\mathcal{#1}}%

\global\long\def\sigfield{\ensuremath{\mathcal{F}}}%
\global\long\def\sigfieldg{\ensuremath{\mathcal{G}}}%

\global\long\def\flt#1{\mathcal{#1}}%

\global\long\def\filt{\mathcal{F}}%
\global\long\def\filtg{\mathcal{G}}%

\global\long\def\Borel{\ensuremath{\mathcal{B}}}%

\global\long\def\cyl{\ensuremath{\mathcal{C}}}%

\global\long\def\nulls{\ensuremath{\mathcal{N}}}%

\global\long\def\gauss{\mathfrak{g}}%

\global\long\def\leb{\mathfrak{m}}%

% Probability and expectation

\global\long\def\prob{\ensuremath{\mathbb{P}}}%

\global\long\def\Prob{\ensuremath{\mathbb{P}}}%

\global\long\def\Probs{\mathcal{P}}%

\global\long\def\PROBS{\mathcal{M}}%

\global\long\def\expect{\ensuremath{\mathbb{E}}}%

\global\long\def\Expect{\ensuremath{\mathbb{E}}}%

\global\long\def\probspc{\ensuremath{(\smpl,\filt,\Prob)}}%

% Distributions and stochastic convergence

\global\long\def\iid{\ensuremath{\textnormal{i.i.d.}}}%

\global\long\def\as{\ensuremath{\textnormal{a.s.}}}%

\global\long\def\asp{\ensuremath{\textnormal{a.s.p.}}}%

\global\long\def\io{\ensuremath{\ensuremath{\textnormal{i.o.}}}}%

\newcommand\independent{\protect\mathpalette{\protect\independenT}{\perp}}
\def\independenT#1#2{\mathrel{\rlap{$#1#2$}\mkern2mu{#1#2}}}

\global\long\def\indep{\independent}%

\global\long\def\distrib{\ensuremath{\sim}}%

\global\long\def\distiid{\ensuremath{\sim_{\iid}}}%

\global\long\def\asydist{\ensuremath{\overset{a}{\distrib}}}%

\global\long\def\inprob{\ensuremath{\overset{p}{\goesto}}}%

\global\long\def\inprobu#1{\ensuremath{\overset{#1}{\goesto}}}%

\global\long\def\inas{\ensuremath{\overset{\as}{\goesto}}}%

\global\long\def\eqas{=_{\as}}%

\global\long\def\inLp#1{\ensuremath{\overset{\Ellp{#1}}{\goesto}}}%

\global\long\def\indist{\ensuremath{\overset{d}{\goesto}}}%

\global\long\def\eqdist{=_{d}}%

\global\long\def\wkc{\ensuremath{\rightsquigarrow}}%

\global\long\def\wkcu#1{\overset{#1}{\ensuremath{\rightsquigarrow}}}%

\global\long\def\plim{\operatorname*{plim}}%

% Moments

\global\long\def\var{\operatorname{var}}%

\global\long\def\lrvar{\operatorname{lrvar}}%

\global\long\def\cov{\operatorname{cov}}%

\global\long\def\corr{\operatorname{corr}}%

\global\long\def\bias{\operatorname{bias}}%

\global\long\def\MSE{\operatorname{MSE}}%

\global\long\def\med{\operatorname{med}}%

\global\long\def\avar{\operatorname{avar}}%

\global\long\def\se{\operatorname{se}}%

\global\long\def\sd{\operatorname{sd}}%

% Testing

\global\long\def\nullhyp{H_{0}}%

\global\long\def\althyp{H_{1}}%

\global\long\def\ci{\mathcal{C}}%

% STOCHASTIC PROCESSES

\global\long\def\simple{\mathcal{R}}%

\global\long\def\sring{\mathcal{A}}%

\global\long\def\sproc{\mathcal{H}}%

\global\long\def\Wiener{\ensuremath{\mathbb{W}}}%

\global\long\def\sint{\bullet}%

\global\long\def\cv#1{\left\langle #1\right\rangle }%

\global\long\def\smlcv#1{\langle#1\rangle}%

\global\long\def\qv#1{\left[#1\right]}%

\global\long\def\smlqv#1{[#1]}%

% MISCELLANEOUS

\global\long\def\trans{\mathsf{T}}%

\global\long\def\indic{\ensuremath{\mathbf{1}}}%

\global\long\def\Lagr{\mathcal{L}}%

\global\long\def\grad{\nabla}%

\global\long\def\pmin{\ensuremath{\wedge}}%
\global\long\def\Pmin{\ensuremath{\bigwedge}}%

\global\long\def\pmax{\ensuremath{\vee}}%
\global\long\def\Pmax{\ensuremath{\bigvee}}%

\global\long\def\sgn{\operatorname{sgn}}%

\global\long\def\argmin{\operatorname*{argmin}}%

\global\long\def\argmax{\operatorname*{argmax}}%

\global\long\def\Rp{\operatorname{Re}}%

\global\long\def\Ip{\operatorname{Im}}%

\global\long\def\deriv{\ensuremath{\mathrm{d}}}%

\global\long\def\diffnspc{\ensuremath{\deriv}}%

\global\long\def\diff{\ensuremath{\,\deriv}}%

\global\long\def\i{\ensuremath{\mathrm{i}}}%

\global\long\def\e{\mathrm{e}}%

\global\long\def\sep{,\ }%

\global\long\def\defeq{\coloneqq}%

\global\long\def\eqdef{\eqqcolon}%

%% CROSS REFERENCING

\begin{comment}
\global\long\def\objlabel#1#2{\{\}}%

\global\long\def\objref#1{\{\}}%
\end{comment}

\global\long\def\err{\varepsilon}%

\global\long\def\mset#1{\mathcal{#1}}%

\global\long\def\largedec#1{\mathbf{#1}}%

\global\long\def\z{\largedec z}%

\newcommandx\Ican[1][usedefault, addprefix=\global, 1=]{I_{#1}^{\ast}}%

\global\long\def\jsr{{\scriptstyle \mathrm{JSR}}}%

\global\long\def\cjsr{{\scriptstyle \mathrm{CJSR}}}%

\global\long\def\rsr{{\scriptstyle \mathrm{RJSR}}}%

\global\long\def\ctspc{\mathscr{M}}%

\global\long\def\b#1{\boldsymbol{#1}}%

\global\long\def\pseudy{\text{\ensuremath{\abv y}}}%

\global\long\def\regcoef{\kappa}%

\global\long\def\adj{\operatorname{adj}}%

\global\long\def\llangle{\dbllangle}%

\global\long\def\rrangle{\dblrangle}%

\global\long\def\smldblangle#1{\ensuremath{\llangle#1\rrangle}}%

\newcommand{\casecens}{{\upshape{(i)}}}

\newcommand{\caseclas}{{\upshape{(ii)}}}

\newcommand{\casestat}{{\upshape{(iii)}}}

\global\long\def\delcens{\mathrm{(i)}}%

\global\long\def\delclas{\mathrm{(ii)}}%

\global\long\def\orthm{\mathcal{O}}%

\global\long\def\loc{d}%

\global\long\def\vekh{\operatorname{vech}}%

\title{Stationarity with Occasionally Binding Constraints}
\author{\authaffil{James A.\ Duffy\footnotemark[1]{}}{University of Oxford}\hspace{1cm}
\authaffil{Sophocles Mavroeidis\footnotemark[2]{}}{University
of Oxford} \hspace{1cm} \authaffil{Sam Wycherley\footnotemark[3]{}}{Stanford
University}}
\date{\vspace*{0.3cm}November 2026}

\maketitle
\renewcommand*{\thefootnote}{\fnsymbol{footnote}}

\footnotetext[1]{Department\ of Economics and Corpus Christi College;
\texttt{james.duffy@economics.ox.ac.uk}}

\footnotetext[2]{Department of Economics and University College;
\texttt{sophocles.mavroeidis@economics.ox.ac.uk}}

\footnotetext[3]{Department of Economics; \texttt{wycherley@stanford.edu}}

\renewcommand*{\thefootnote}{\arabic{footnote}}

\setcounter{footnote}{0}
\begin{abstract}
\noindent This paper studies a class of multivariate threshold autoregressive
models, known as censored and kinked structural vector autoregressions
(CKSVAR), which are notably able to accommodate series that are subject
to occasionally binding constraints. We develop a set of sufficient
conditions for the processes generated by a CKSVAR to be stationary,
ergodic, and weakly dependent. Our conditions relate directly to the
stability of the deterministic part of the model, and are therefore
less conservative than those typically available for general vector
threshold autoregressive (VTAR) models. Though our criteria refer
to quantities, such as refinements of the joint spectral radius, that
cannot feasibly be computed exactly, they can be approximated numerically
to a high degree of precision.
\end{abstract}
\vfill
\begin{elabeling}{00.00.0000}
\item [{Keywords:}] structural VAR; regime switching; threshold autoregression;
Markov process; geometric ergodicity; joint spectral radius
\item [{MSC codes:}] 60G10, 60G65, 60J05, 62M10
\item [{JEL codes:}] C32
\end{elabeling}
\vspace{0.5cm}

\noindent{}The results of this paper were first presented in Sections~2
and 3 of arXiv:2211.09604v1. We thank participants at seminars at
Cambridge, Cyprus, Stanford and Oxford, for comments on earlier drafts
of this work.

\thispagestyle{plain}

\pagenumbering{roman}

\newpage{}

\thispagestyle{plain}

\setcounter{tocdepth}{2}

\tableofcontents{}

\newpage{}

\pagenumbering{arabic}

\section{Introduction}

This paper studies a class of multivariate threshold autoregressive
models known as censored and kinked structural vector autoregressions
(CKSVAR; \citealp{SM21}). These models feature endogenous regime
switching induced by threshold-type nonlinearities, i.e.\ changes
in coefficients and variances when one of the variables crosses a
threshold, but differ importantly from a previous generation of vector
threshold autoregressive models (VTAR; see e.g.\ \citealp{TTG10})
insofar as the autoregressive `regime' is determined endogenously,
rather than being predetermined. This notably allows the CKSVAR model
to accommodate series that are subject to occasionally binding constraints,
a leading example of which is provided by the zero lower bound (ZLB)
on short-term nominal interest rates \citep{SM21,AMSV21,ILMZ20}.

In this paper, we develop a set of sufficient conditions for the processes
generated by the CKSVAR to be stationary, ergodic, and weakly dependent.
These conditions are of interest, firstly, because the credibility
of a structural model relies on its being able to generate plausible
counterfactual trajectories and associated impulse responses, i.e.\ on
its being adequate to replicate the most elementary time-series properties
of the data. If that data is apparently stationary, model configurations
that instead give rise to explosive trajectories and responses should
naturally be excluded from the parameter space (in a Bayesian setting,
by assigning zero prior mass to these regions; cf.\ \citealp{AMSV21},
Sec.\ 3.3). Secondly, frequentist inference on the parameters of
these models relies on the series generated by the model, under the
null hypothesis of interest, satisfying the requirements of laws of
large numbers and central limit theorems for dependent data (\citealp{Wool94Hdbk};
\citealp{PP97}).

Because of the nonlinearity of the CKSVAR, the stationarity of the
model defies the simple analytical characterisation that is available
in a linear VAR. (In particular, it is insufficient to simply check
the magnitudes of the autoregressive roots associated with some or
all of the autoregressive `regimes' implied by the model: see \exaref{explosive}
below.) This is a problem routinely encountered in the literature
on nonlinear time-series models, and some of the approaches taken
in that literature may be fruitfully applied here. Specifically, we
rely on existing results from the theory of ergodic Markov processes
(\citealp{Tjo90AAP}; \citealp{MT09}), casting the CKSVAR as an instance
of the `regime-switching' or VTAR models considered in that literature
(see e.g.\ \citealp{Tong90}; \citealp{Chan09}; \citealp{TTG10};
\citealp{HT13}), in order to establish conditions sufficient for
the series generated by a CKSVAR to be stationary, geometrically ergodic,
and $\beta$-mixing (absolutely regular).

Within the broad class of VTAR models, the CKSVAR enjoys some additional
structure that we exploit to obtain criteria for stationarity that
are less conservative than the criteria applicable to general VTAR
models. In particular, because the CKSVAR is continuous at the threshold
(kink) and approximately homogeneous (of degree one),\footnote{We say that a function $f:\reals^{m}\setmap\reals^{n}$ is \emph{homogeneous
of degree one} if $f(tx)=tf(x)$ for all $t>0$ and $x\in\reals^{m}$.
Note in particular the requirement that $t$ should be strictly positive.} its stationarity can be characterised directly in terms of the stability
of the deterministic part of the model, as per \citet{CT85AAP}. By
contrast, VTAR models are generally either discontinuous at their
thresholds, or -- when those discontinuities are smoothed away --
cannot be approximated by a model that is \emph{both} continuous and
homogeneous of degree one, precluding the approach taken in this paper.

We develop a hierarchy of sufficient conditions for stability, only
the most elementary of which, the joint spectral radius (JSR), has
previously been applied to nonlinear time-series models (e.g.\ \citealp{Lieb2005};
\citealp{MS08JTSA}; \citealp{Saik08ET}; \citealp{KS2020ER}). The
JSR may be regarded as adopting a `worst case' perspective, which
guarantees stability even under arbitrary sequences of autoregressive
regimes, i.e.\ without requiring those successive regimes to be chosen
according to the model's internal dynamics. It may thus be applied
to models lacking the continuity or approximate homogeneity characteristic
of the CKSVAR -- as in the papers cited -- but being conservative
may fail to correctly classify parameters that are consistent with
stability. In such cases, our less restrictive criteria -- the \emph{constrained}
or \emph{relaxed }JSR -- may nevertheless succeed in classifying
the model as stable, as illustrated by the examples in \subsecref{numerical-examples}
below. Though our criteria refer to quantities that cannot feasibly
be computed exactly, they can be approximated numerically to a high
degree of precision (arbitrarily well, given sufficient computation
time, in some cases); we have developed the \R{} package \texttt{thresholdr}
to provide users of these models with a means of numerically verifying
these criteria.\footnote{Available at: \url{https://github.com/samwycherley/thresholdr}}
Beyond the setting of the CKSVAR, these criteria should also be useful
for determining the stationarity of other nonlinear time-series models
in which stationarity can be referred to the stability of the deterministic
part of the model (such as is the case for any model satisfying the
conditions of our \lemref{ergodic}).

While the present work is primarily concerned with stationarity and
ergodicity, one may also be interested in conditions under which the
CKSVAR is able to accommodate series with multiple stochastic trends.
Unlike in a linear VAR, where stochastic trends arise by allowing
some autoregressive roots to equal (real) unity, the nonlinearity
of the CKSVAR prevents such trends from being treated as mere boundary
cases of the stationary model (in which e.g.\ the JSR now equals
unity, rather than being bounded strictly below unity). The successful
extension of the Granger--Johansen representation theorem to this
setting therefore requires some novel arguments. That extension, which
shows that the CKSVAR can generate both linear and nonlinear forms
of cointegration, is given in a companion paper to the present work
\citep{DMW22}, utilising the canonical representation of the CKSVAR
model introduced in \propref{canonical} below.

The remainder of the paper is organised as follows. \secref{model}
introduces the CKSVAR model and develops its canonical representation.
In \secref{stationarity}, we provide sufficient conditions for the
CKSVAR to generate stationary, ergodic and weakly dependent time series,
and show that these conditions also deliver consistency and asymptotic
normality of maximum likelihood estimators of the model parameters.
Those sufficient conditions depend, in turn, on the stability of the
deterministic part of the model, criteria for which are elaborated
in \secref{stability}.

\begin{notation*}
$e_{m,i}$ denotes the $i$th column of an $m\times m$ identity matrix
$I_{m}$; when $m$ is clear from the context, we write this simply
as $e_{i}$. In a statement such as $f(a^{\pm},b^{\pm})=0$, the notation
`$\pm$' signifies that both $f(a^{+},b^{+})=0$ and $f(a^{-},b^{-})=0$
hold; similarly, `$a^{\pm}\in A$' denotes that both $a^{+}$ and
$a^{-}$ are elements of $A$. All limits are taken as $n\goesto\infty$
unless otherwise stated. $\inprob$ and $\wkc$ respectively denote
convergence in probability and in distribution (weak convergence).
$\smlnorm{\cdot}$ denotes the Euclidean norm on $\reals^{m}$; all
matrix norms are induced by the corresponding vector norms. For $X$
a random vector and $p\geq1$, $\smlnorm X_{p}\defeq(\expect\smlnorm X^{p})^{1/p}$.
\end{notation*}

\section{Model}

\label{sec:model}

\subsection{The censored and kinked SVAR}

We consider a VAR($k$) model in $p$ variables, in which one series,
$y_{t}$, enters with coefficients that differ according to whether
it is above or below a time-invariant threshold $b$, while the other
$p-1$ series, collected in $x_{t}$, enter linearly. Define
\begin{align}
y_{t}^{+} & \defeq\max\{y_{t},b\} & y_{t}^{-} & \defeq\min\{y_{t},b\}\label{eq:y-threshold}
\end{align}
and consider the model
\begin{equation}
\phi_{0}^{+}y_{t}^{+}+\phi_{0}^{-}y_{t}^{-}+\Phi_{0}^{x}x_{t}=c+\sum_{i=1}^{k}[\phi_{i}^{+}y_{t-i}^{+}+\phi_{i}^{-}y_{t-i}^{-}+\Phi_{i}^{x}x_{t-i}]+u_{t}\label{eq:var-two-sided}
\end{equation}
which we may write more compactly as
\begin{equation}
\phi^{+}(L)y_{t}^{+}+\phi^{-}(L)y_{t}^{-}+\Phi^{x}(L)x_{t}=c+u_{t},\label{eq:var-pm}
\end{equation}
where $\phi^{\pm}(L)\defeq\phi_{0}^{\pm}-\sum_{i=1}^{k}\phi_{i}^{\pm}L^{i}$
and $\Phi^{x}(L)\defeq\Phi_{0}^{x}-\sum_{i=1}^{k}\Phi_{i}^{x}L^{i}$,
for $\phi_{i}^{\pm}\in\reals^{p\times1}$ and $\Phi_{i}^{x}\in\reals^{p\times(p-1)}$,
and $L$ denotes the lag operator. Here $\{u_{t}\}$ is a sequence
of $p$-dimensional innovations, whose properties will be further
specified in \assref{err} below. If $b\neq0$, then by defining $y_{b,t}\defeq y_{t}-b$,
$y_{b,t}^{+}\defeq\max\{y_{b,t},0\}$, $y_{b,t}^{-}\defeq\min\{y_{b,t},0\}$
and $c_{b}\defeq c-[\phi^{+}(1)+\phi^{-}(1)]b$, we may rewrite \eqref{var-pm}
as
\begin{equation}
\phi^{+}(L)y_{b,t}^{+}+\phi^{-}(L)y_{b,t}^{-}+\Phi^{x}(L)x_{t}=c_{b}+u_{t}.\label{eq:zero-thresh}
\end{equation}
For the purposes of this paper, we may thus take $b=0$ without loss
of generality. In this case, $y_{t}^{+}$ and $y_{t}^{-}$ respectively
equal the positive and negative parts of $y_{t}$, and $y_{t}=y_{t}^{+}+y_{t}^{-}$.
(Throughout the following, the notation `$a^{\pm}$' connotes $a^{+}$
and $a^{-}$ as objects associated respectively with $y_{t}^{+}$
and $y_{t}^{-}$, or their lags. If we want to instead denote the
positive and negative parts of some $a\in\reals$, we shall do so
by writing $[a]_{+}\defeq\max\{a,0\}$ or $[a]_{-}\defeq\min\{a,0\}$.)

Models of the form \eqref{var-pm} have previously been employed in
the literature to account for the dynamic effects of censoring, occasionally
binding constraints, and endogenous regime switching: see \citet{SM21},
\citet{AMSV21}, \citet{ABH23mimeo}, \citet{ILMZ20}, \citet{CCMM25};
and for extensions of the model considered here, see \citet{DM24,DM26id}.\footnote{The model formulated by \citet{AMSV21} falls within the scope of
\eqref{var-pm}, once the conditions necessary for their model to
have a unique solution (for all values of $u_{t}$) are imposed: see
their Proposition~1(i).} We shall follow \citet{SM21} in referring to \eqref{var-pm} as
a \emph{censored and kinked structural VAR} (CKSVAR) model, even though
in many applications of interest $y_{t}^{+}$ and $y_{t}^{-}$ are
fully observed, i.e.\ there is no `censoring'. The following running
example will be used to illustrate the concepts developed in this
paper.\footnote{In contrast to \eqref{var-two-sided}, which allows the two components
$y_{t}^{+}$ and $y_{t}^{-}$ of the endogenous variable $y_{t}$
to have different coefficients, \citet{DG94JF} and \citet{JdG21CS}
consider models of the form
\[
z_{t}=\sum_{i=1}^{p}\Phi_{i}z_{t-i}+\sum_{i=0}^{q}(B_{i}^{+}u_{t-i}^{+}+B_{i}^{-}u_{t-i}^{-})
\]
in which innovations of opposite signs are permitted to have different
moving average coefficients. Since the nonlinearity is entirely confined
to those moving average terms, the conditions for stationarity in
these models are simply the usual ones for the stationarity of linear
vector autoregressive models.}
\begin{example}[monetary policy]
\label{exa:monetary} Consider the following stylised structural
model of monetary policy in the presence of a zero lower bound (ZLB)
constraint on interest rates consisting of a composite IS and Phillips
curve (PC) equation
\begin{align}
\pi_{t}-\abv{\pi} & =\chi(\pi_{t-1}-\abv{\pi})+\theta[i_{t}^{+}+\mu i_{t}^{-}-(r_{t}^{\ast}+\abv{\pi})]+\varepsilon_{t}\label{eq:is-pc}
\end{align}
and a policy reaction function (Taylor rule)
\begin{equation}
i_{t}=(r_{t}^{\ast}+\abv{\pi})+\gamma(\pi_{t}-\abv{\pi})\label{eq:taylor}
\end{equation}
where $r_{t}^{\ast}$ denotes the (real) natural rate of interest,
$\pi_{t}$ inflation, and $\varepsilon_{t}$ a mean zero, i.i.d.\ innovation.
The stance of monetary policy is measured by $i_{t}$: with $i_{t}^{+}=[i_{t}]_{+}$
giving the actual policy rate (constrained to be non-negative), and
$i_{t}^{-}=[i_{t}]_{-}$ the desired stance of policy when the policy
rate is constrained by the ZLB, to be effected via some form of `unconventional'
monetary policy, such as long-term asset purchases. $\abv{\pi}$ denotes
the central bank's inflation target. We maintain that $\gamma>0$,
$\theta<0$, $\chi\in[0,1)$, and $\mu\in[0,1]$, where this last
parameter captures the relative efficacy of unconventional policy.
When $\chi=0$, the preceding corresponds to a simplified version
of the model of \citet{ILMZ20}; a model with $\chi>0$ emerges by
augmenting their Phillips curve with a measure of backward-looking
agents.

To `close' the model, we specify that the natural real rate of interest
follows a stationary AR(1) process,
\begin{equation}
r_{t}^{\ast}=(1-\psi)\abv r+\psi r_{t-1}^{\ast}+\eta_{t}\label{eq:natural-rate}
\end{equation}
where $\psi\in(-1,1)$, and $\eta_{t}$ is an i.i.d.\ mean zero innovation,
possibly correlated with $\varepsilon_{t}$ (cf.\ \citealp{LW03REStat},
for a model in which $\psi=1$, as is also considered in \citealp{DMW22}).
By substituting \eqref{taylor} into \eqref{is-pc} and \eqref{natural-rate},
and letting $\abv{\imath}\defeq\abv r+\abv{\pi}$, we obtain
\begin{equation}
\begin{bmatrix}1 & 1 & -\gamma\\
0 & \theta(1-\mu) & 1-\theta\gamma
\end{bmatrix}\begin{bmatrix}i_{t}^{+}\\
i_{t}^{-}\\
\pi_{t}
\end{bmatrix}=\begin{bmatrix}(1-\psi)(\abv{\imath}-\gamma\abv{\pi})\\
(1-\theta\gamma-\chi)\abv{\pi}
\end{bmatrix}+\begin{bmatrix}\psi & \psi & -\psi\gamma\\
0 & 0 & \chi
\end{bmatrix}\begin{bmatrix}i_{t-1}^{+}\\
i_{t-1}^{-}\\
\pi_{t-1}
\end{bmatrix}+\begin{bmatrix}\eta_{t}\\
\varepsilon_{t}
\end{bmatrix},\label{eq:cksvar-natrate}
\end{equation}
rendering the system as a CKSVAR for $(i_{t},\pi_{t})$.
\end{example}

The CKSVAR encompasses both kinds of dynamic Tobit model as special
cases (for applications of which, in both time series and panel settings,
see e.g.\ \citealp{DJ02FRB}; \citealp{DJH11}; \citealp{DSK12AE};
\citealp{LMS19}; \citealp{BMMV21JBF}; and \citealp{Byk21JBES}).
\begin{example}[univariate]
\label{exa:univariate}Consider \eqref{var-pm} with $p=1$ and $\phi_{0}^{+}=\phi_{0}^{-}=1$,
so that
\begin{equation}
y_{t}=c+\sum_{i=1}^{k}(\phi_{i}^{+}y_{t-i}^{+}+\phi_{i}^{-}y_{t-i}^{-})+u_{t},\label{eq:univariate-case}
\end{equation}
and suppose only $y_{t}^{+}$ is observed. In the nomenclature of
\citet[Sec.~1]{BD22}, if $\phi_{i}^{-}=0$ for all $i\in\{1,\ldots,k\}$,
so that only the positive part of $y_{t-i}$ enters the r.h.s., then
\begin{equation}
y_{t}^{+}=\left[c+\sum_{i=1}^{k}\phi_{i}^{+}y_{t-i}^{+}+u_{t}\right]_{+}\label{eq:censored-Tobit}
\end{equation}
follows a \emph{censored} dynamic Tobit; whereas if $\phi_{i}^{+}=\phi_{i}^{-}=\phi_{i}$
for all $i\in\{1,\ldots,k\}$, then $y_{t}^{+}$ follows a \emph{latent}
dynamic Tobit, being simply the positive part of the linear autoregression,
\begin{equation}
y_{t}=c+\sum_{i=1}^{k}\phi_{i}y_{t-i}+u_{t}\label{eq:latent-Tobit}
\end{equation}
(cf.\ \citealt[p.~186]{Maddala83}; \citealt[p.~419]{wei1999}).
\end{example}
As discussed by \citet{SM21} and \citet{AMSV21}, the model \eqref{var-pm}
is not guaranteed to have a unique solution for $(y_{t},x_{t})$,
at least not for all possible values of $u_{t}$, unless certain conditions
are placed on entries of the matrix
\[
\Phi_{0}\defeq\begin{bmatrix}\phi_{0}^{+} & \phi_{0}^{-} & \Phi_{0}^{x}\end{bmatrix}=\begin{bmatrix}\phi_{0,yy}^{+} & \phi_{0,yy}^{-} & \phi_{0,yx}^{\trans}\\
\phi_{0,xy}^{+} & \phi_{0,xy}^{-} & \Phi_{0,xx}
\end{bmatrix}
\]
of contemporaneous coefficients; or equivalently on the matrices
\begin{align*}
\Phi_{0}^{+} & \defeq\begin{bmatrix}\phi_{0}^{+} & \Phi_{0}^{x}\end{bmatrix} & \Phi_{0}^{-} & \defeq\begin{bmatrix}\phi_{0}^{-} & \Phi_{0}^{x}\end{bmatrix}
\end{align*}
that respectively apply when $y_{t}$ is positive or negative. By
\citet[Prop.~1]{SM21}, \eqref{var-pm} has a unique solution --
the model is `coherent' (see also \citealp{GLM80Ecta}) -- if the
second condition of the following holds.

\needspace{3\baselineskip}

\assumpname{DGP}
\begin{assumption}
\label{ass:dgp}~
\begin{enumerate}[label=\ass{\arabic*.}, ref=\ass{.\arabic*}]
\item \label{enu:dgp:defn} $\{(y_{t},x_{t})\}$ are generated according
to \eqref{y-threshold}--\eqref{var-pm} with $b=0$, with possibly
random initial values $(y_{i},x_{i})$, for $i\in\{-k+1,\ldots,0\}$;
\item \label{enu:dgp:coherence} $\sgn(\det\Phi_{0}^{+})=\sgn(\det\Phi_{0}^{-})\neq0$.
\item \label{enu:dgp:wlog} $\Phi_{0,xx}$ is invertible, and
\begin{equation}
\sgn\{\phi_{0,yy}^{+}-\phi_{0,yx}^{\trans}\Phi_{0,xx}^{-1}\phi_{0,xy}^{+}\}=\sgn\{\phi_{0,yy}^{-}-\phi_{0,yx}^{\trans}\Phi_{0,xx}^{-1}\phi_{0,xy}^{-}\}>0.\label{eq:coherence}
\end{equation}
\end{enumerate}
\end{assumption}
\begin{rem}
\label{rem:dgp}\subremark{} Under \ref{ass:dgp}\enuref{dgp:coherence},
a reordering of the equations in \eqref{var-pm}, and thus of the
rows of $\Phi_{0}$, ensures $\Phi_{0,xx}$ is invertible. Since
$\phi_{0,yy}^{\pm}-\phi_{0,yx}^{\trans}\Phi_{0,xx}^{-1}\phi_{0,xy}^{\pm}=\det(\Phi_{0}^{\pm})/\det(\Phi_{0,xx})$,
the equality in \eqref{coherence} holds; the inequality can thus
be satisfied by multiplying \eqref{var-pm} through by $\pm1$ as
appropriate. Thus \ref{ass:dgp}\enuref{dgp:wlog} is without loss
of generality, and we have stated it here only to clarify that we
shall treat these normalisations as holding, purely for the sake of
convenience, throughout the sequel.

\subremark{} An important special case of the CKSVAR arises when
$y_{t}^{-}$ is not observed -- or when it is only observed up to
scale -- in which case there is a continuum of observationally equivalent
parametrisations of \eqref{var-pm}, each of which generates identical
time series for the observables $(y_{t}^{+},x_{t})$, but in which
the trajectories for $y_{t}^{-}$ are scaled by some constant. (In
e.g.\ \citealp{SM21}, this arises because the `shadow rate' is
unobservable below zero.) In this case, the scale of the coefficients
$\{\phi_{i}^{-}\}_{i=0}^{k}$ on $y_{t}^{-}$ in \eqref{var-pm} can
be normalised in a convenient way that permits certain simplifications
to be made; we shall refer to this as the case of a \emph{partially
observed} CKSVAR.
\end{rem}

The following provides an illustration of the role of these coherency
conditions, and motivates the use of the CKSVAR as a model of endogenous
regime switching (in which $z_{t}$ is fully observed in both regimes).
\begin{example}[state-dependent fiscal multipliers]
\label{exa:fiscal-multipliers} \citet{AG12AEJ} and \citet{RZ18JPE}
investigate whether the effects of fiscal policy may be `state dependent',
by allowing impulse responses to depend on some (predetermined) measure
of economic slack. In the latter paper, the authors estimate these
impulse responses via a local projection (i.e.\ a regression) of
the form
\begin{align}
x_{t+h} & =\b 1_{t-1}\left[\alpha_{A,h}+\omega_{A,h}(L)q_{t-1}+\beta_{A,h}s_{t}\right]\nonumber \\
 & \qquad\qquad\qquad+(1-\b 1_{t-1})\left[\alpha_{B,h}+\omega_{B,h}(L)q_{t-1}+\beta_{B,h}s_{t}\right]+\varepsilon_{t+h}\label{eq:fiscal-lp}
\end{align}
where $s_{t}$ is an (observed) fiscal shock, $\omega_{A,h}(L)$ is
a row vector of autoregressive polynomials, $q_{t-1}$ is a conformable
vector of predetermined variables and $\b 1_{t-1}$ indicates that
the economy is in the `slack' regime immediately prior to the realisation
of $s_{t}$. More concretely, the authors (following \citealp{ORZ13AER})
define the slack regime as one in which unemployment $u_{t}$ exceeds
some threshold $\bar{u}$ (which they set equal to $6.5$ per cent),
so that taking $y_{t}\defeq u_{t}-\bar{u}$ we have $\b 1_{t-1}\defeq\indic\{y_{t-1}>0\}$.

Since the state is predetermined -- i.e.\ determined prior to the
realisation of $s_{t}$ -- the specification \eqref{fiscal-lp} precludes
the shock itself moving the economy between the two regimes \emph{on
impact}. This is an inherent limitation of the local projection methodology
used in these papers, and also of the STVAR model utilised by \citet{AG12AEJ},
which may be overcome within the CKSVAR framework. To give a stylised
illustration of how this may be approached, let us now suppose that
the regime is determined \emph{contemporaneously} with the realisation
of $s_{t}$. Collecting any predetermined and exogenous terms into
$w_{t}$, this suggests a specification for unemployment $y_{t}=u_{t}-\bar{u}$
(expressed in deviations from the threshold $\bar{u}$) at horizon
zero ($h=0$) of the form
\begin{equation}
y_{t}=\begin{cases}
\beta_{s}^{+}s_{t}+w_{t}, & \text{if }y_{t}>0,\\
\beta_{s}^{-}s_{t}+\beta_{w}^{-}w_{t}, & \text{if }y_{t}\leq0.
\end{cases}\label{eq:fiscal-incoherent}
\end{equation}
where we have implicitly defined $w_{t}$ such that $\beta_{w}^{+}=1$.
This model is only `coherent' (in the sense of \citealp{GLM80Ecta})
if it yields a unique solution for $y_{t}$ for all possible values
of $s_{t}$ and $w_{t}$. Letting $g_{t}^{U}\defeq\beta_{s}^{+}s_{t}+w_{t}$
and $g_{t}^{L}\defeq\beta_{s}^{-}s_{t}+\beta_{w}^{-}w_{t}$ denote
the two `branches' of the r.h.s., we see that if $g_{t}^{U}\leq0<g_{t}^{L}$,
then neither branch is consistent with its defining inequality, and
so the model admits no solution for $y_{t}$; whereas if these are
both reversed as $g_{t}^{L}\leq0<g_{t}^{U}$, then the model yields
two solutions.

\begin{figure}
\includegraphics[width=1\textwidth]{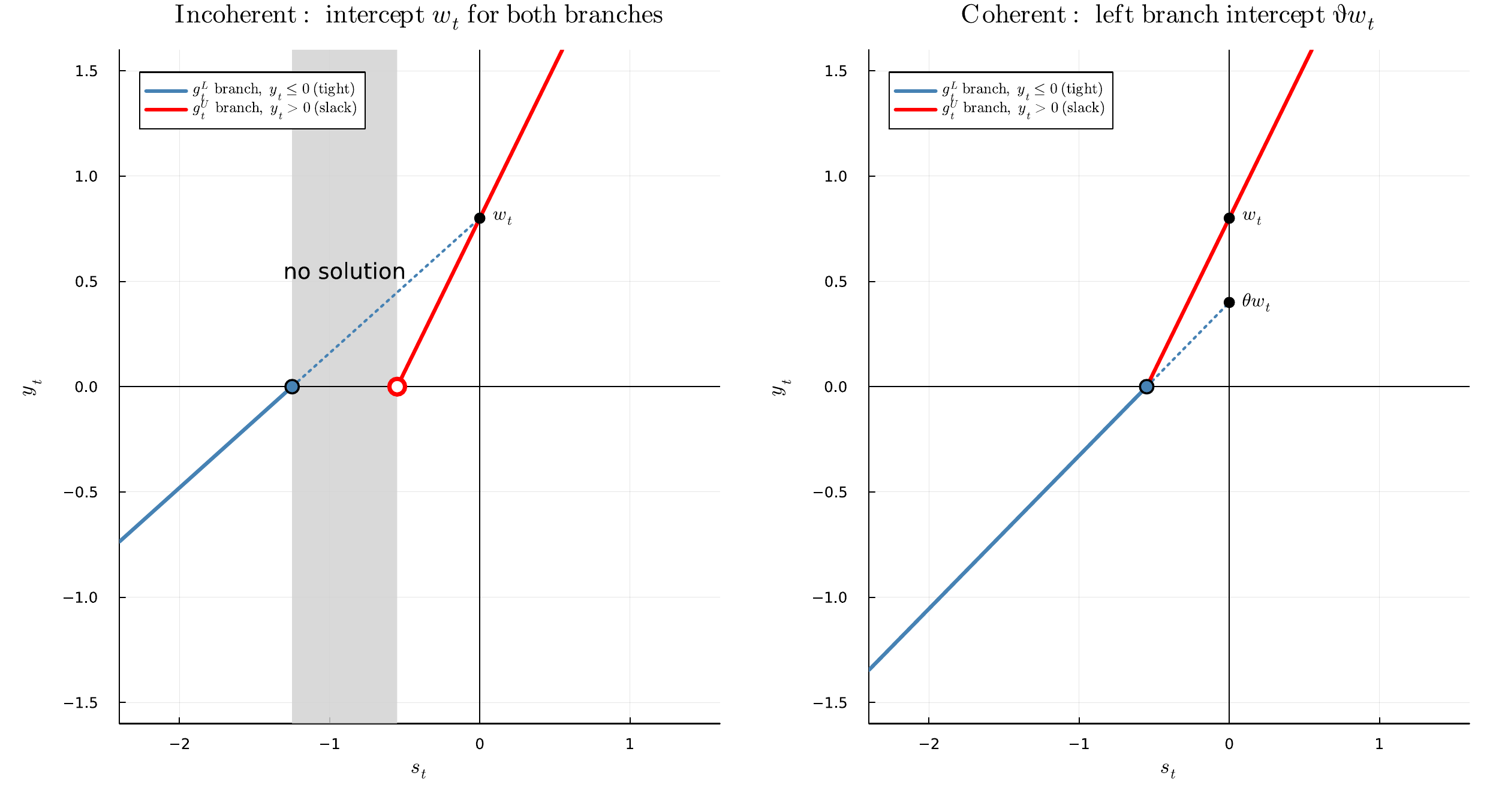}

\caption{\protect\label{fig:fiscal-coherency}In the unrestricted model (left),
the two branches cross zero at different values of $s_{t}$, producing
an interval on which the model yields no solution for $y_{t}$. In
the coherent model (right), the intercept and slope of the lower branch
are both scaled by $\vartheta>0$, so that the branches meet at the
threshold and the model is coherent.}
\end{figure}

Thus \eqref{fiscal-incoherent} does not in general define a coherent
model; this instead requires that one, and only one, of $g_{t}^{U}$
and $g_{t}^{L}$ is consistent with their associated regime inequalities.
However, this will obtain under the additional restrictions that:
(i) $\beta_{s}^{+}$ and $\beta_{s}^{-}$ have the same sign (as in
\ref{ass:dgp}\enuref{dgp:coherence} above), and (ii) the two branches
coincide at the $y_{t}=0$ threshold, i.e.
\[
\beta_{s}^{+}s_{t}+w_{t}=0\iff\beta_{s}^{-}s_{t}+\beta_{w}^{-}w_{t}=0.
\]
The latter is equivalent to $\beta_{w}^{-}=\beta_{s}^{-}/\beta_{s}^{+}\eqdef\vartheta$,
where $\vartheta>0$ under (i), in which case
\[
g_{t}^{L}=\beta_{s}^{-}s_{t}+\beta_{w}^{-}w_{t}=\vartheta(\beta_{s}^{+}s_{t}+w_{t})=\vartheta g_{t}^{U}\eqdef\vartheta g_{t}
\]
Hence under these coherency restrictions, the model \eqref{fiscal-incoherent}
becomes 
\[
y_{t}=\begin{cases}
g_{t}, & \text{if }y_{t}>0,\\
\vartheta g_{t}, & \text{if }y_{t}\leq0;
\end{cases}
\]
for $g_{t}=\beta_{s}^{+}s_{t}+w_{t}$; or equivalently,
\begin{equation}
y_{t}^{+}+\vartheta^{-1}y_{t}^{-}=g_{t},\label{eq:fiscal-coherent}
\end{equation}
which is of precisely the CKSVAR form \eqref{var-two-sided}, with
$\vartheta^{-1}>0$ so that the coherency conditions embedded in \assref{dgp}
are satisfied. \figref{fiscal-coherency} contrasts the unrestricted
and coherent models, i.e.\ \eqref{fiscal-incoherent} and \eqref{fiscal-coherent}.
\end{example}

\subsection{The canonical CKSVAR}

We say that a CKSVAR is \emph{canonical} if
\begin{equation}
\Phi_{0}=\begin{bmatrix}1 & 1 & 0\\
0 & 0 & I_{p-1}
\end{bmatrix}\eqdef\Ican[p].\label{eq:canonical}
\end{equation}
While it is not always the case that the reduced form of \eqref{var-pm}
corresponds directly to a canonical CKSVAR, by defining the canonical
variables
\begin{equation}
\begin{bmatrix}\tilde{y}_{t}^{+}\\
\tilde{y}_{t}^{-}\\
\tilde{x}_{t}
\end{bmatrix}\defeq\begin{bmatrix}\bar{\phi}_{0,yy}^{+} & 0 & 0\\
0 & \bar{\phi}_{0,yy}^{-} & 0\\
\phi_{0,xy}^{+} & \phi_{0,xy}^{-} & \Phi_{0,xx}
\end{bmatrix}\begin{bmatrix}y_{t}^{+}\\
y_{t}^{-}\\
x_{t}
\end{bmatrix}\eqdef P^{-1}\begin{bmatrix}y_{t}^{+}\\
y_{t}^{-}\\
x_{t}
\end{bmatrix},\label{eq:canon-vars}
\end{equation}
where $\bar{\phi}_{0,yy}^{\pm}\defeq\phi_{0,yy}^{\pm}-\phi_{0,yx}^{\trans}\Phi_{0,xx}^{-1}\phi_{0,xy}^{\pm}>0$
and $P^{-1}$ is invertible under \ref{ass:dgp}; and setting
\begin{equation}
\begin{bmatrix}\tilde{\phi}^{+}(L) & \tilde{\phi}^{-}(L) & \tilde{\Phi}^{x}(L)\end{bmatrix}\defeq Q\begin{bmatrix}\phi^{+}(L) & \phi^{-}(L) & \Phi^{x}(L)\end{bmatrix}P,\label{eq:canon-polys}
\end{equation}
where 
\begin{equation}
Q\defeq\begin{bmatrix}1 & -\phi_{0,yx}^{\trans}\Phi_{0,xx}^{-1}\\
0 & I_{p-1}
\end{bmatrix},\label{eq:Q-canon}
\end{equation}
we obtain the following, whose proof appears in \appref{canonical}.

\needspace{3\baselineskip}
\begin{prop}
\label{prop:canonical}Suppose \ref{ass:dgp} holds. Then:
\begin{enumerate}
\item there exist $(\tilde{y}_{t},\tilde{x}_{t})$ such that \eqref{canon-vars}--\eqref{canon-polys}
hold, $\tilde{y}_{t}^{+}=\max\{\tilde{y}_{t},0\}$, $\tilde{y}_{t}^{-}=\min\{\tilde{y}_{t},0\}$
and
\begin{equation}
\tilde{\phi}^{+}(L)\tilde{y}_{t}^{+}+\tilde{\phi}^{-}(L)\tilde{y}_{t}^{-}+\tilde{\Phi}^{x}(L)\tilde{x}_{t}=\tilde{c}+\tilde{u}_{t},\label{eq:tildeVAR}
\end{equation}
is a canonical CKSVAR, where $\tilde{c}=Qc$ and $\tilde{u}_{t}=Qu_{t}$;
and
\item if the CKSVAR for $(y_{t},x_{t})$ is partially observed, $y_{t}^{-}$
may be rescaled such that we may take $\tilde{y}_{t}=y_{t}$.
\end{enumerate}
\end{prop}
The utility of the canonical representation lies in its rendering
all the nonlinearity in the model as a more tractable function of
the \emph{lags} of $\tilde{y}_{t}$ alone. Because the time-series
properties of a general CKSVAR are inherited from its canonical form,
we may restrict attention to canonical CKSVAR models essentially without
loss of generality. More precisely, for convenience we shall state
our results below for canonical CKSVAR models, and then invert the
mapping \eqref{canon-vars}--\eqref{canon-polys} to determine the
implications of these results for general CKSVAR models. That is,
we shall initially maintain the following.

\assumpname{DGP$^{\ast}$}
\begin{assumption}
\label{ass:dgp-canon}$\{(y_{t},x_{t})\}$ are generated by a canonical
CKSVAR, i.e.\ \ref{ass:dgp} holds with $\Phi_{0}=[\phi_{0}^{+},\phi_{0}^{-},\Phi_{0}^{x}]=\Ican[p]$
(and $b=0$); so that \eqref{var-two-sided} may be equivalently rewritten
as 
\begin{equation}
\begin{bmatrix}y_{t}\\
x_{t}
\end{bmatrix}=c+\sum_{i=1}^{k}\begin{bmatrix}\phi_{i}^{+} & \phi_{i}^{-} & \Phi_{i}^{x}\end{bmatrix}\begin{bmatrix}y_{t-i}^{+}\\
y_{t-i}^{-}\\
x_{t-i}
\end{bmatrix}+u_{t}.\label{eq:canon-var}
\end{equation}
\end{assumption}
\saveexamplex{}

\exname{\ref*{exa:monetary}}
\begin{example}[monetary policy; ctd]
 From \eqref{cksvar-natrate} above we have
\begin{align*}
\Phi_{0}^{+} & =\begin{bmatrix}1 & -\gamma\\
0 & 1-\theta\gamma
\end{bmatrix} & \Phi_{0}^{-} & =\begin{bmatrix}1 & -\gamma\\
\theta(1-\mu) & 1-\theta\gamma
\end{bmatrix},
\end{align*}
and thus $\det\Phi_{0}^{+}=1-\theta\gamma$ and $\det\Phi_{0}^{-}=1-\mu\theta\gamma$,
both of which are positive, so that the coherency condition \assref{dgp}\enuref{dgp:coherence}
is satisfied. (The model is already written in a form such that \assref{dgp}\enuref{dgp:wlog}
also holds.) The model may be rendered in canonical form by taking
\begin{gather}
\begin{aligned}\begin{bmatrix}\tilde{\imath}_{t}^{+}\\
\tilde{\imath}_{t}^{-}\\
\tilde{\pi}_{t}
\end{bmatrix} & \defeq\begin{bmatrix}1 & 0 & 0\\
0 & 1+\tau_{\mu} & 0\\
0 & \theta(1-\mu) & 1-\theta\gamma
\end{bmatrix}\begin{bmatrix}i_{t}^{+}\\
i_{t}^{-}\\
\pi_{t}
\end{bmatrix}\qquad & \qquad\begin{bmatrix}\tilde{\eta}_{t}\\
\tilde{\varepsilon}_{t}
\end{bmatrix} & \defeq\begin{bmatrix}1 & \gamma(1-\theta\gamma)^{-1}\\
0 & 1
\end{bmatrix}\begin{bmatrix}\eta_{t}\\
\varepsilon_{t}
\end{bmatrix}\end{aligned}
\label{eq:ex-canon-var}\\
\tilde{c}\defeq\begin{bmatrix}1 & \gamma(1-\theta\gamma)^{-1}\\
0 & 1
\end{bmatrix}\begin{bmatrix}(1-\psi)(\abv{\imath}-\gamma\abv{\pi})\\
(1-\theta\gamma-\chi)\abv{\pi}
\end{bmatrix}
\end{gather}
where $\tau_{\mu}\defeq\gamma\theta(1-\mu)(1-\theta\gamma)^{-1}$,
for which it holds that
\begin{equation}
\begin{bmatrix}\tilde{\imath}_{t}\\
\tilde{\pi}_{t}
\end{bmatrix}=\tilde{c}+\begin{bmatrix}\psi & \psi-\chi\tau_{\mu}\kappa_{\mu} & \gamma(\chi\kappa_{1}-\psi)\kappa_{1}\\
0 & -\chi\theta(1-\mu)\kappa_{\mu} & \chi\kappa_{1}
\end{bmatrix}\begin{bmatrix}\tilde{\imath}_{t-1}^{+}\\
\tilde{\imath}_{t-1}^{-}\\
\tilde{\pi}_{t-1}
\end{bmatrix}+\begin{bmatrix}\tilde{\eta}_{t}\\
\tilde{\varepsilon}_{t}
\end{bmatrix}\label{eq:example-canon}
\end{equation}
where $\kappa_{\mu}\defeq(1-\mu\theta\gamma)^{-1}$. (In the special
case where $\chi=0$, the canonical form is a linear system, since
then both $\tilde{\imath}_{t-1}^{+}$ and $\tilde{\imath}_{t-1}^{-}$
enter the r.h.s.\ with the same coefficients.)
\end{example}
\restoreexamplex{}

\section{Stationarity and ergodicity}

\label{sec:stationarity}

\subsection{The CKSVAR as a regime-switching model}

\label{subsec:regimeswitching}

The nonlinearity of the CKSVAR makes the assessment of its stationarity
a rather more complicated affair than it is for a linear (S)VAR. Though
it may be natural to think of the model as having two autoregressive
`regimes', corresponding respectively to positive and negative values
of $\{y_{t}\}$, and associated with the autoregressive polynomials
\[
\Phi^{\pm}(L)\defeq\begin{bmatrix}\phi^{\pm}(L) & \Phi^{x}(L)\end{bmatrix},
\]
the roots of $\Phi^{\pm}(L)$ are not sufficient to characterise the
stationarity of the model. What seems an intuitive criterion for stationarity,
that all these roots should lie outside the unit circle, turns out
to be neither necessary nor sufficient -- at best, it is merely necessary
for \emph{another} criterion for stability to be satisfied, as we
develop below. Even in very special cases, such as a CKSVAR in which
only lags of $y_{t}^{+}$ enter the model (as e.g.\ in \citealp{AMSV21}),
it is possible to construct numerical examples of the following kind.

\begin{figure}
\begin{adjustwidth}{-2cm}{-2cm}
\begin{centering}
\begin{tabular}{cc}
\includegraphics[viewport=230bp 165bp 612bp 440bp,clip,scale=0.6]{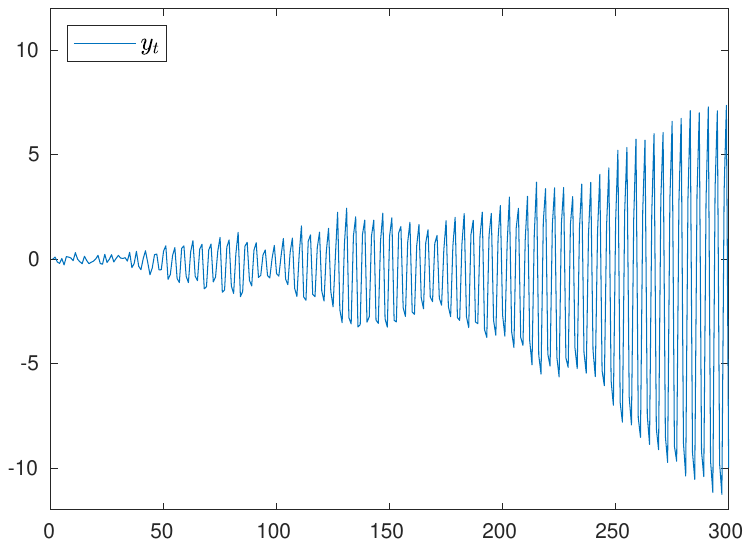} & \includegraphics[viewport=230bp 165bp 612bp 440bp,clip,scale=0.6]{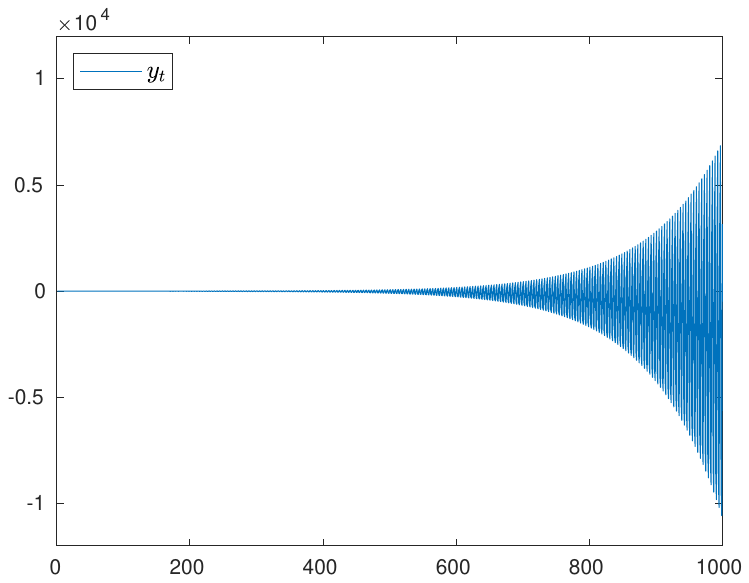}\tabularnewline
\end{tabular}
\par\end{centering}
\caption{Simulated trajectories for the CKSVAR in \exaref{explosive}}
\label{fig:unstable}

\end{adjustwidth}
\end{figure}

\begin{example}[explosive]
\label{exa:explosive} Consider a CKSVAR with $p=3$ and $k=1$,
where $\Phi_{0}=\Ican[3]$ and
\begin{align*}
\phi_{1}^{+} & =\begin{bmatrix}-1.37\\
0.79\\
0.76
\end{bmatrix} & \phi_{1}^{-} & =\begin{bmatrix}0\\
0\\
0
\end{bmatrix} & \Phi_{1}^{x} & =\begin{bmatrix}-1.00 & 0.36\\
0.39 & -1.33\\
0.71 & 0.03
\end{bmatrix},
\end{align*}
and $c=0$. The model has only two autoregressive regimes, associated
with the autoregressive matrices $\Phi^{+}=[\phi_{1}^{+},\Phi_{1}^{x}]$
and $\Phi^{-}=[\phi_{1}^{-},\Phi_{1}^{x}]$; it may be verified that
the eigenvalues of each lie inside the unit circle, with the largest
(in modulus) being $0.85$ and $0.98$ respectively. Nonetheless,
simulated trajectories generated by this model are clearly explosive,
as illustrated by \figref{unstable}, which plots $\{y_{t}\}$ generated
under $\Sigma=I_{3}$.

To account for this evident instability, we may observe that the deterministic
part of the model has an explosive cycle, corresponding to the sign
pattern $\{+,+,-,-\}$ for $y_{t}$. This corresponds to a product
$(\Phi^{-})^{2}(\Phi^{+})^{2}$ of successive autoregressive matrices,
which has a largest eigenvalue of approximately $1.04$, and an associated
eigenvector $v\approx(0.27,-0.96,0.02)^{\trans}$. Starting from $z_{t}=v$,
so that $y_{t}>0$, yields the sequence
\begin{align*}
z_{t+1} & =\Phi^{+}v, & z_{t+2} & =(\Phi^{+})^{2}v, & z_{t+3} & =\Phi^{-}(\Phi^{+})^{2}v, & z_{t+4} & =(\Phi^{-})^{2}(\Phi^{+})^{2}v=\lambda v
\end{align*}
where indeed $y_{t+1}>0$, $y_{t+2}<0$ and $y_{t+3}<0$. Along this
path, $z_{t}$ expands exponentially at rate $\lambda^{1/4}$. The
simulations depicted in \figref{unstable} illustrate that this explosive
path is dominant even in the presence of shocks, with the indicated
sign pattern for $y_{t}$ evident from the left panel.
\end{example}
As the above example illustrates, a tractable analytic characterisation
of stationarity is unlikely to be available in this setting. To make
progress we draw on approaches developed in the literature on nonlinear
time-series models, specifically those pertaining to `regime-switching'
or `vector threshold autoregressive' (VTAR) models (see \citealp{HT13}),
of which the CKSVAR is an instance. To make the connections with these
models more apparent, let $\indic^{+}(y)\defeq\indic\{y\geq0\}$ and
$\indic^{-}(y)\defeq\indic\{y<0\}$, and define\footnote{There is unavoidably some arbitrariness with respect to how these
objects are defined when $y=0$, but since these only play a role
in the model when multiplied by $y$, it does not matter what convention
is adopted. Throughout the paper, we use the notation $\indic^{\pm}(y)$
to ensure that all such definitions are mutually consistent.} 
\begin{align*}
\phi_{i}(y) & \defeq\phi_{i}^{+}\indic^{+}(y)+\phi_{i}^{-}\indic^{-}(y) & \Phi_{i}(y) & \defeq\begin{bmatrix}\phi_{i}(y) & \Phi_{i}^{x}\end{bmatrix}
\end{align*}
for $i\in\{1,\ldots,k\}$. Then for $z_{t}\defeq(y_{t},x_{t}^{\trans})^{\trans}$,
the canonical CKSVAR
\begin{equation}
z_{t}=c+\sum_{i=1}^{k}\Phi_{i}(y_{t-i})z_{t-i}+u_{t}\label{eq:regimeswitching}
\end{equation}
may be regarded as an autoregressive model with $2^{k}$ `regimes'
corresponding to all the possible sign patterns of $(y_{t-1},\ldots,y_{t-k})$,
with switches between those regimes occurring at each of the $k$
`thresholds' defined by $y_{t-i}=0$. While the existing literature
provides results on the stationarity of a general class of regime-switching
vector autoregressive models (see e.g.\ \citealp{Lieb2005}; \citealp{Saik08ET};
\citealp{KS2020ER}), because of that very generality we are able
to improve on those results by exploiting the special structure of
the CKSVAR model.

\subsection{Ergodicity via stability of the associated deterministic system}

\label{subsec:ergodicstable}

To state our results, we first recall the following (cf.\ \citealp{Lieb2005},
p.\ 671).
\begin{defn}
Let $\{w_{t}\}_{t\in\naturals_{0}}$ be a Markov chain taking values
in $\reals^{d_{w}}$, with $m$-step transition kernel $P^{m}(w,A)\defeq\Prob\{w_{t+m}\in A\mid w_{t}=w\}$,
and $\mathcal{Q}:\reals^{d_{w}}\setmap\reals_{+}$. We say that $\{w_{t}\}$
is $\mathcal{Q}$-\emph{geometrically ergodic}, with stationary distribution
$\pi$, if $\int_{\reals^{d_{w}}}\mathcal{Q}(w)\pi(\deriv w)<\infty$,
and there exist $a,b>0$ and $\gamma\in(0,1)$ such that
\[
\sup_{B\in\Borel}\smlabs{P^{m}(w,B)-\pi(B)}\leq(a+b\mathcal{Q}(w))\gamma^{m}
\]
for all $w\in\reals^{d_{w}}$, where $\Borel$ denotes the Borel sigma-field
on $\reals^{d_{w}}$.
\end{defn}
If $\{w_{t}\}$ is $\mathcal{Q}$-geometrically ergodic, then it is
positive Harris recurrent, a property that is particularly useful
for establishing the large-sample properties of estimators of the
CKSVAR parameters, as discussed in \subsecref{asymptotics} below.
Moreover $\{w_{t}\}$ will be stationary and ergodic if given a stationary
initialisation, i.e.\ if $w_{0}$ also has distribution $\pi$, and
will have geometrically decaying $\beta$-mixing coefficients. (For
these and further properties of such sequences, see \citealp{Lieb2005},
pp.~671--73, and also Ch.\ 9 of \citealp{MT09}, for a discussion
of Harris recurrence.)

While $\{z_{t}\}$ is not a Markov chain, $\b z_{t}\defeq(z_{t}^{\trans},\ldots,z_{t-k+1}^{\trans})^{\trans}$
is, as can be seen by rendering \eqref{regimeswitching} in companion
form. To establish that $\{\b z_{t}\}$ is $\mathcal{Q}$-geometrically
ergodic, we shall make the following assumption on the innovation
sequence $\{u_{t}\}$ in \eqref{var-pm}, which allows for a certain
form of conditional heteroskedasticity.

\needspace{3\baselineskip}

\assumpname{ERR}
\begin{assumption}
\label{ass:err}$u_{t}=\Sigma(\b z_{t-1})\err_{t}$, where
\begin{enumerate}[label=\ass{\arabic*.}, ref=\ass{.\arabic*}]
\item \label{enu:err:leb}$\{\err_{t}\}$ is i.i.d.\ and independent of
$\b z_{t-1}$, with a Lebesgue density that is bounded away from zero
on compact subsets of $\reals^{p}$, and $\expect\smlnorm{\err_{t}}^{m_{0}}<\infty$
for some $m_{0}\geq1$;
\item \label{enu:err:eigs}for each compact $K\subset\reals^{kp}$, there
exist $c_{0},c_{1}\in(0,\infty)$ such that the singular values of
$\Sigma(\b z)$ lie in $[c_{0},c_{1}]$ for all $\b z\in K$; and
\item \label{enu:err:sigmaterms}$\Sigma_{ij}(\b z)=o(\smlnorm{\b z})$
as $\smlnorm{\b z}\goesto\infty$, for each $i,j\in\{1\ldots,p\}$.
\end{enumerate}
\end{assumption}
\begin{rem}
\ref{ass:err}\enuref{err:leb} is a standard assumption in this setting,
ensuring that the chain is irreducible and that its stationary distribution
is continuous (cf.\ \citealp[p.~675]{Lieb2005}, or \citealp[Ass.~1]{Saik08ET}).
\assref{err}\ref{enu:err:eigs} and \assref{err}\ref{enu:err:sigmaterms}
are respectively equivalent to (11) and (10) in \citet{Lieb2005}.
\end{rem}
We may now state our main result on the geometric ergodicity of the
CKSVAR, which relates the geometric ergodicity of $\{\b z_{t}\}$
to the stability of the deterministic system (or `skeleton' in the
terminology of \citealp{Tong90}) associated to \eqref{regimeswitching},
defined by
\begin{equation}
\hat{z}_{t}=\sum_{i=1}^{k}\Phi_{i}(\hat{y}_{t-i})\hat{z}_{t-i}.\label{eq:determ}
\end{equation}
We say that such a system is \emph{stable} if $\hat{z}_{t}\goesto0$
as $t\goesto\infty$, for every initialisation $\{\hat{z}_{t}\}_{t=-k+1}^{0}$.
Our results for the CKSVAR follow as a corollary to a more general
result for continuous and homogeneous (of degree one) autoregressive
systems, given as \lemref{ergodic} in \appref{ergodic}. This in
turn extends a fundamental result due to \citet[Thms~4.2 and 4.5]{CT85AAP}
to the present setting, principally by relaxing their conditions on
the innovations $\{u_{t}\}$.\footnote{Cf.\ Theorem~3.1 in \citet{CP99StSin}, which also relaxes these
conditions, but yields only the weaker conclusion of geometric ergodicity;
whereas $\mathcal{Q}$-ergodicity additionally ensures that the stationary
distribution has finite $m_{0}$th moment.} The proof of the following appears in \appref{ergodic}.

\begin{thm}
\label{thm:ergodicity} Suppose \assref{dgp-canon} and \assref{err}
hold, and that the deterministic system \eqref{determ} associated
to the model \eqref{regimeswitching} is stable. Then $\{\b z_{t}\}_{t\in\naturals_{0}}$
is $\mathcal{Q}$-geometrically ergodic, for $\mathcal{Q}(\b z)\defeq1+\smlnorm{\b z}^{m_{0}}$,
with a stationary distribution that is equivalent to Lebesgue measure
on $\reals^{kp}$, and has finite $m_{0}$th moment. If in addition
$\expect\smlnorm{\b z_{0}}^{m_{0}}<\infty$, then $\sup_{t\in\naturals}\expect\smlnorm{\b z_{t}}^{m_{0}}<\infty$.
\end{thm}

For a process $\{z_{t}\}$ generated by a general CKSVAR, we note
that \propref{canonical} yields a bi-Lipschitz mapping between $z_{t}$
and the associated canonical process $\tilde{z}_{t}$. Since stability
of (the deterministic part of) the CKSVAR implies stability of its
canonical form, we have the following.
\begin{cor}
\label{cor:ergodicity}Suppose \assref{dgp} and \assref{err} hold,
and that the associated deterministic system
\begin{equation}
\Phi_{0}(\hat{y}_{t})\hat{z}_{t}=\sum_{i=1}^{k}\Phi_{i}(\hat{y}_{t-i})\hat{z}_{t-i}\label{eq:deterCKSVAR}
\end{equation}
is stable. Then the conclusions of \thmref{ergodicity} hold.
\end{cor}

\begin{rem}
\label{rem:ergod}\subremark{} In view of \lemref{ergodic} in \appref{ergodic},
the preceding would also hold if the CKSVAR model were generalised
by allowing the intercept to vary as $c=c(\b z_{t-1})$, provided
that $c=o(\smlnorm{\b z})$ as $\smlnorm{\b z}\goesto\infty$.

\subremark{} These results hold trivially in a linear VAR, being
a special case of the CKSVAR. However, whereas the stability of a
linear VAR can be determined from its autoregressive roots, assessing
the stability of a CKSVAR requires more elaborate criteria, such as
those developed in \secref{stability} below.

\subremark{}\label{subrem:order} The preceding holds more generally
when $\b z_{t}=(z_{t}^{\trans},\ldots,z_{t-k+1}^{\trans})^{\trans}$
is replaced by $\b z_{t}^{(k^{\prime})}\defeq(z_{t}^{\trans},\ldots,z_{t-k^{\prime}+1}^{\trans})^{\trans}$
for any $k^{\prime}>k$, as can be seen by nesting a (stable) $k$th-order
model within a (stable) $k^{\prime}$th-order model, in which the
parameters referring to the final $k^{\prime}-k$ lags of $z_{t}$
are identically zero. Regarding the initialisation of $\b z_{t}^{(k^{\prime})}$,
observe that if $\expect\smlnorm{\b z_{0}}^{m_{0}}<\infty$, then
$\expect\smlnorm{z_{t}}^{m_{0}}<\infty$ for $t\in\{1,\ldots,k^{\prime}-k\}$,
and thus we can regard $\{\b z_{t}^{(k^{\prime})}\}$ as being initialised
at time $t=k^{\prime}-k$ with $\expect\smlnorm{\b z_{k^{\prime}-k}^{(k^{\prime})}}^{m_{0}}<\infty$.
\end{rem}

\subsection{Implications for asymptotic inference}

\label{subsec:asymptotics}

\thmref{ergodicity} and its corollary provide sufficient conditions
for the Markov process generated by a CKSVAR to be $\mathcal{Q}$-geometrically
ergodic, and therefore also positive Harris recurrent. This permits
the application of laws of large numbers for such processes, facilitating
the derivation of the limiting distributions of estimators for the
CKSVAR model \eqref{var-two-sided}. By way of illustration, here
we provide a result on the asymptotic normality of maximum likelihood
estimators (MLE) under the simplifying assumptions of: i.i.d.\ Gaussian
innovations; a known threshold (here normalised to zero); and $z_{t}$
being fully observed in both regimes.\footnote{If instead only $y_{t}^{+}$ is observed, so that $y_{t}^{-}$ must
be treated as a latent process, the problem of evaluating the likelihood
becomes much more challenging: see Appendix~E of \citet{SM21} for
details. This also significantly complicates the asymptotic analysis
of the MLE, something that is the subject of the authors' ongoing
research.} To fix the overall scale of the model parameters, we also normalise
$\Sigma=I_{p}$ -- as is also appropriate if the elements of $u_{t}$
are to be interpreted as (mutually uncorrelated) structural shocks
(see e.g.\ \citealp{KL17book}).

Because the CKSVAR is a (dynamic) system of (nonlinear) simultaneous
equations, data on $\{z_{t}\}$ alone is insufficient to uniquely
identify the values of the structural parameters $\b{\Phi}\defeq(c,\{\phi_{i}^{+},\phi_{i}^{-},\Phi_{i}^{x}\}_{i=0}^{k})$.
Indeed, \citet{DM26id} show that when $\{z_{t}\}$ is fully observed,
the model \eqref{var-two-sided} falls within a broad class of `endogenously
nonlinear' SVARs whose parameters are identified up to some (unknown)
orthogonal matrix $\orthm\in\reals^{p\times p}$. Theirs is a nonparametric
identification result, which requires only weak smoothness conditions
on the distribution of the innovations, and which a fortiori holds
also when those innovations are assumed to belong to a parametric
family that is closed under orthogonal transformations, such as the
Gaussian. The parameters of the model are thus identified, from the
data $\{z_{t}\}$, to the same extent as they would be in a linear
SVAR.

It is accordingly helpful to reparametrise the CKSVAR model so as
to clearly separate the identified parameters from those that are
unidentified. (This provides an alternative to the analysis of an
identified `reduced form' model, an approach which is far less convenient
in this setting than in a linear SVAR, due to the possible nonlinearity
on the l.h.s.\ of \eqref{var-two-sided}.) Since $\Phi_{0}^{+}=[\phi_{0}^{+},\Phi_{0}^{x}]$
is invertible, it uniquely decomposes (via the QL decomposition) as
\[
[\phi_{0}^{+},\Phi_{0}^{x}]=\Phi_{0}^{+}=\orthm^{\trans}\Psi_{0}^{+}=\orthm^{\trans}[\psi_{0}^{+},\Psi_{0}^{x}]
\]
where $\orthm\in\reals^{p\times p}$ is an orthogonal matrix, and
$\Psi_{0}^{+}\in\reals^{p\times p}$ is lower triangular, normalised
to have positive diagonal entries. (So that $\sgn\det\Psi_{0}^{+}=1$
and, as a consequence of which, \assref{dgp} entails that $\sgn\det\Psi_{0}^{-}=1$
also. It follows that the diagonal entries of $\Psi_{0}^{-}$ are
also all positive.) Multiplying \eqref{var-two-sided} through by
$\orthm$ yields the reparametrised model
\begin{equation}
\psi_{0}^{+}y_{t}^{+}+\psi_{0}^{-}y_{t}^{-}+\Psi_{0}^{x}x_{t}=\mu+\sum_{i=1}^{k}[\psi_{i}^{+}y_{t-i}^{+}+\psi_{i}^{-}y_{t-i}^{-}+\Psi_{i}^{x}x_{t-i}]+\orthm u_{t}.\label{eq:reparm}
\end{equation}
By Theorem~2.2 of \citet{DM26id}, the parameters $\b{\Psi}\defeq(\mu,\{\psi_{i}^{+},\psi_{i}^{-},\Psi_{i}^{x}\}_{i=0}^{k})$
are exactly identified, whereas $\orthm$ is not. However, since $\orthm u_{t}\distiid N[0,I_{p}]$,
the model likelihood for $\b{\Psi}$ is invariant to $\orthm$, and
thus the former may be estimated (by ML) without regard for the latter.
The result below establishes that these MLEs are consistent and asymptotically
normal -- and in view of the form of the limiting distribution, asymptotically
efficient.

Inference on the structural parameters $\b{\Phi}$ would require external
identifying restrictions sufficient to pin down (some or all of) $\orthm$,
such that the elements of $\b{\Phi}$ may be recovered from those
of $\b{\Psi}$ via 
\[
\orthm^{\trans}\begin{bmatrix}\psi_{i}^{+} & \psi_{i}^{-} & \Psi_{i}^{x}\end{bmatrix}=\begin{bmatrix}\phi_{i}^{+} & \phi_{i}^{-} & \Phi_{i}^{x}\end{bmatrix}.
\]
Such restrictions may be obtained, exactly as they are in a linear
SVAR, by appealing to macroeconomic theories that either constrain
the associated structural impulse response functions (e.g.\ by imposing
that certain impact multipliers are zero, or restricting their signs),
or which suggest external instruments that are correlated only with
certain structural shocks: see \citet{DM26id} for a discussion. 

Let $\theta$ collect all elements of $\mu$, $\psi_{0}^{+}$, $\psi_{0}^{-}$,
$\{\psi_{i}^{+},\psi_{i}^{-},\Psi_{i}^{x}\}_{i=1}^{k}$, and the below-diagonal
elements of $\Psi_{0}^{x}$ (recall that $\Psi_{0}^{x}$ is zero on
and above its main diagonal). The parameter space $\Theta\subset\reals^{d_{\theta}}$
for $\theta$ is restricted only by the requirement that $\Psi_{0}^{\pm}$
have positive diagonal entries; we let $\theta_{0}\in\Theta$ denote
the true parameter under which the data $\{z_{t}\}$ was generated.
In view of \eqref{reparm}, when $u_{t}\distiid N[0,I_{p}]$ the model
implies a conditional density $f_{\theta}(\bar{z}_{t}\mid\bar{\b z}_{t-1})$
for the observables that depends on $\theta$ but not on $\orthm$.
(Here $\bar{z}_{t}$ and $\bar{\b z}_{t-1}$ denote real values that
may be taken by the random variables $z_{t}$ and $\b z_{t-1}$.)
Let $\hat{\theta}_{n}$ denote a maximiser of the loglikelihood $L_{n}(\theta)\defeq\sum_{t=1}^{n}\log f_{\theta}(z_{t}\mid\b z_{t-1})$
conditional on the initialisation $\b z_{0}=(z_{0},\ldots,z_{-k+1})$,
and define the associated (conditional) score process $s_{t}\defeq\grad_{\theta}\left.\log f_{\theta}(z_{t}\mid\b z_{t-1})\right|_{\theta=\theta_{0}}$,
where $\grad_{\theta}\left.g(\theta)\right|_{\theta=\theta_{0}}$
denotes the gradient of $g(\theta)$ at $\theta=\theta_{0}$. Under
the conditions of \corref{ergodicity}, the Markov process $(z_{t}^{\trans},\b z_{t-1}^{\trans})^{\trans}$
will converge weakly to its invariant distribution as $t\goesto\infty$;
let $\mathbf{E}$ denote an expectation taken with respect to that
distribution. The proof of the following result appears in \appref{limit-dist}.

\needspace{3\baselineskip}
\begin{thm}
\label{thm:limittheory}Suppose that \assref{dgp} holds, the deterministic
system \eqref{deterCKSVAR} is stable, $u_{t}\distiid N[0,I_{p}]$,
$\expect\smlnorm{\b z_{0}}^{m_{0}}<\infty$ for some $m_{0}>4$, and
$\theta_{0}\in\intr\Theta$. Then
\[
n^{1/2}(\hat{\theta}_{n}-\theta_{0})\wkc N[0\sep\mathcal{I}(\theta_{0})^{-1}],
\]
where $\mathcal{I}(\theta_{0})\defeq\plim_{n\goesto\infty}\frac{1}{n}\sum_{t=1}^{n}s_{t}s_{t}^{\trans}=\mathbf{E}s_{1}s_{1}^{\trans}$
is positive definite.
\end{thm}
\begin{rem}
\subremark{} By Theorem~2.2 of \citet{DM26id}, the parameters $\b{\Psi}$
are identified even if $\{u_{t}\}$ is merely assumed to be i.i.d.\ with
a (Lebesgue) density that satisfies some weak smoothness conditions,
and which has full support on $\reals^{p}$. Nonetheless, it does
not seem possible to estimate $\b{\Psi}$ semiparametrically, e.g.\ by
OLS regression, because of the potential nonlinearity on the l.h.s.\ of
\eqref{reparm}. One might instead take a semi-nonparametric approach
to the estimation of these parameters, in which the distribution of
the identified shocks $\orthm u_{t}$ (as opposed to $u_{t}$ itself)
were approximated by sieves, the parameters regulating which would
be estimated (by ML) jointly with $\b{\Psi}$.

\subremark{} Inspection of the proof shows that \thmref{limittheory}
generalises straightforwardly to the case where $\{u_{t}\}$ is i.i.d.\ with
$\expect u_{t}=0$, $\expect u_{t}u_{t}^{\trans}=I_{p}$, and has
a strictly log-concave density $f$ that is invariant to orthogonal
transformations. This holds, for example, if $f(u)=C\exp[-g(\smlnorm u)]$
for some $g:\reals_{+}\setmap\reals_{+}$ that is strictly convex
and increasing. In this case, the convexity argument used in the proof
to deduce the limiting distribution of $\hat{\theta}_{n}$ from the
finite-dimensional convergence of the (rescaled) loglikelihood function
may still be applied.
\end{rem}

\section{Sufficient conditions for stability}

\label{sec:stability}

While the preceding reduces the problem to one of verifying the stability
of the deterministic subsystem \eqref{determ}, this still presents
a formidable challenge. We therefore develop some sufficient conditions
for stability -- and hence for ergodicity -- that can be evaluated
numerically. Some of these concepts, particularly the joint spectral
radius (JSR), have been previously applied to regime-switching and
VTAR models, and to this extent there is an overlap between our results
and those of \citet[Thm.~2]{Lieb2005} and \citet[Thm.~1]{KS2020ER}.
However, by exploiting the close connection between the ergodicity
of a CKSVAR and the stability of its associated deterministic system,
as per \thmref{ergodicity}, we can assess ergodicity using less conservative
criteria than are available when working with a broader class of regime-switching
autoregressive processes.

\subsection{In the abstract}

\label{subsec:stabilityabstract}

To formulate the criteria in abstract terms, let $\{w_{t}\}_{t\in\naturals_{0}}$
be a deterministic process, taking values in $\reals^{d_{w}}$, that
evolves according to a \emph{switched linear system}: that is, a system
in which there is a (possibly uncountable) set ${\cal I}\ni i$ of
states, each of which is associated with a distinct autoregressive
matrix $A[i]$. If $\{i_{t}\}\subset{\cal I}$ records the state of
the system at each $t\in\naturals_{0}$, then $\{w_{t}\}$ evolves
as
\begin{equation}
w_{t}=A[i_{t-1}]w_{t-1}\label{eq:unresdet}
\end{equation}
from some given $w_{0}\in\reals^{d_{w}}$. Suppose further that the
state in period $t$ is entirely determined by the value of $w_{t}$,
via the mapping $\sigma:\reals^{d_{w}}\setmap{\cal I}$. Then $i_{t}=\sigma(w_{t})$,
and
\begin{equation}
w_{t}=A[\sigma(w_{t-1})]w_{t-1}\eqdef A(w_{t-1})w_{t-1}\label{eq:detsystem}
\end{equation}
Let ${\cal A}\defeq\{A[i]\}_{i\in{\cal I}}$ denote the set of autoregressive
matrices, and $\mathscr{W}_{i}\defeq\sigma^{-1}(i)$, so that $\{\mathscr{W}_{i}\}_{i\in{\cal I}}$
partitions $\reals^{d_{w}}$.

We say \eqref{detsystem} is \emph{stable} if $w_{t}\goesto0$ as
$t\goesto\infty$, for every $w_{0}\in\reals^{d_{w}}$ (what is more
precisely termed `global asymptotic stability'; e.g.\ \citealp[pp.~176f.]{Ela05}).
The usual approach to establishing the stability of such a system
is to construct a Lyapunov function. For our purposes, we may take
this to be a function $V$ and an associated value $\gamma_{V}\in\reals_{+}$
such that
\begin{enumerate}
\item $c_{0}\smlnorm w\leq V(w)$ for all $w\in\reals^{d_{w}}$, for some
$c_{0}\in(0,\infty)$; and
\item $V[A(w)w]\leq\gamma_{V}V(w)$ for all $w\in\reals^{d_{w}}$.
\end{enumerate}
As is well known, if we can find a $(V,\gamma_{V})$ with $\gamma_{V}<1$,
the system is stable (\citealp[Thm.~4.2]{Ela05}): and thus the problem
of establishing the stability of \eqref{detsystem} can be reformulated
as one of showing that the \emph{stability degree} 
\[
\gamma^{\ast}\defeq\inf\{\gamma_{V}\in\reals_{+}\mid(V,\gamma_{V})\text{ satisfy (i)--(ii)}\}
\]
of the system is strictly less than unity. (While stability may be
established using a Lyapunov function satisfying weaker conditions
than (i)--(ii), for the systems considered here these conditions
are not restrictive: see \lemref{lyapunov}.)

While the calculation of $\gamma^{\ast}$ is in general an undecidable
problem, various upper bounds for it are known. The most elementary
of these is provided by the joint spectral radius, which for a bounded
collection of matrices ${\cal A}$ may be defined as (see e.g.\ \citealp{Jungers09},
Defn.\ 1.1)
\[
\rho_{\jsr}(\mathcal{A})\defeq\limsup_{t\goesto\infty}\sup_{B\in\mathcal{A}^{t}}\rho(B)^{1/t}
\]
where $\rho(B)$ denotes the spectral radius of $B$, and $\mathcal{A}^{t}\defeq\{\prod_{s=1}^{t}M_{s}\mid M_{s}\in\mathcal{A}\}$
is the collection of all possible $t$-fold products of matrices in
${\cal A}$. For each $\epsilon>0$, there exists a norm $\smlnorm{\cdot}_{\ast}$
such that 
\begin{equation}
\smlnorm{A[i]w}_{\ast}\leq[\rho_{\jsr}(\mathcal{A})+\epsilon]\smlnorm w_{\ast}\sep\forall w\in\reals^{d_{w}}\sep\forall i\in{\cal I},\label{eq:jsrnormbound}
\end{equation}
(see e.g.\ \citealp{Jungers09}, Prop.\ 1.4) whence $V(w)\defeq\smlnorm w_{\ast}$
trivially yields a valid Lyapunov function, and so $\gamma^{\ast}\leq\rho_{\jsr}(\mathcal{A})$.
Note that $\sup_{i\in{\cal I}}\rho(A[i])$ provides only a \emph{lower}
bound on $\rho_{\jsr}(\mathcal{A})$, so control over the individual
spectral radii is necessary, but not sufficient, for control over
the JSR of ${\cal A}$.

The conservativeness of the JSR is readily apparent from the fact
that it implies the stability of \eqref{unresdet} \emph{irrespective}
of the sequence $\{i_{t}\}$, so that a system adjudged to be stable
by this criterion would remain so even if $\sigma$ in \eqref{detsystem}
were replaced by an arbitrary mapping $\sigma^{\prime}:\reals^{d_{w}}\setmap{\cal I}$.
The JSR can thus be immediately improved upon by taking account of
the constraints on the sequence $\{i_{t}\}$ of states permitted by
the system. Suppose now that ${\cal I}$ is finite, and let ${\cal J}\subset{\cal I}\times{\cal I}$
denote the set of all pairs $(i,j)$ such that state $i$ may be followed
by state $j$, i.e.\ $(i,j)\in{\cal J}$ if there exists a $w\in\mathscr{W}_{i}$
such that $A[\sigma(w)]w\in\mathscr{W}_{j}$; we say that $\{i_{t}\}$
is \emph{admissible} \emph{by }${\cal J}$ if $(i_{t},i_{t+1})\in{\cal J}$
for all $t$. Then the \emph{constrained} \emph{joint spectral radius}
(CJSR) may be defined as (cf.\ \citealp{Dai12LAA}, Defn.\ 1.2 and
Thm.~A; \citealp{PEDJ16Auto}, p.\ 243)
\begin{equation}
\rho_{\cjsr}(A[\cdot];{\cal J})\defeq\limsup_{t\goesto\infty}\sup_{B\in\mathcal{A}^{t}({\cal J})}\rho(B)^{1/t}\label{eq:cjsr}
\end{equation}
where now $\mathcal{A}^{t}({\cal J})\defeq\{\prod_{s=1}^{t}A[i_{s}]\mid\{i_{s}\}\text{ is admissible by }\mathcal{J}\}$.
By \citet[Prop.~2.2]{PEDJ16Auto}, for each $\epsilon>0$ there exists
a family of norms $\{\smlnorm{\cdot}_{i}\}_{i\in\mathcal{I}}$ such
that
\begin{equation}
\smlnorm{A[i]w}_{j}\leq[\rho_{\cjsr}(\mathcal{A};{\cal J})+\epsilon]\smlnorm w_{i}\sep\forall w\in\reals^{d_{w}}\sep\forall(i,j)\in{\cal J}\label{eq:cjsr-norms}
\end{equation}
whence $V(w)\defeq\sum_{i\in\mathcal{I}}\smlnorm w_{i}\indic\{w\in\mathscr{W}_{i}\}$
is a Lyapunov function; it is evident that $\rho_{\cjsr}(A[\cdot];{\cal J})\leq\rho_{\jsr}(\mathcal{A})$,
so that the CJSR yields an improved estimate of the stability degree
of the system.

The form of the Lyapunov function implicitly provided by the CJSR
in turn suggests how this construction may be further improved upon.
For $V(w)=\sum_{i\in\mathcal{I}}\smlnorm w_{i}\indic\{w\in\mathscr{W}_{i}\}$
to be a Lyapunov function, it is sufficient that: (i) such an inequality
as \eqref{cjsr-norms} hold only for $w\in\mathscr{W}_{i}$ such that
$A[i]w\in\mathscr{W}_{j}$, rather than for all $w\in\reals^{d_{w}}$;
and (ii) each $\smlnorm{\cdot}_{i}$ satisfy $\smlnorm w_{i}>c_{i}\smlnorm w$
on $w\in\mathscr{W}_{i}$, for some $c_{i}>0$, i.e.\ $\smlnorm{\cdot}_{i}$
need not itself be a norm. Relaxing \eqref{cjsr-norms} in this manner,
and replacing each norm $\smlnorm{\cdot}_{i}$ by some mapping $\smldblangle{\cdot}_{i}:\reals^{d_{w}}\setmap\reals$
from a class of functions ${\cal C}$, leads us to define the\emph{
relaxed joint spectral radius} (RJSR) \emph{for class }${\cal C}$
as\footnote{When ${\cal C}$ is the set of all norms on $\reals^{d_{w}}$, and
the `$\forall w\in\mathscr{W}_{i}$, etc.'\ qualifiers are replaced
by `$\forall w\in\reals^{d_{w}}$', \eqref{RSR} provides a valid
characterisation of the CJSR: see \citet[Prop.~2.2]{PEDJ16Auto}.}
\begin{align}
\rho_{\rsr,{\cal C}}(A[\cdot];\mathcal{J},\{\mathscr{W}_{i}\}) & \defeq\inf\{\gamma\in\reals_{+}\mid\exists\{\smldblangle{\cdot}_{i},c_{i}\}_{i\in{\cal I}}\ \text{s.t.}\,\smldblangle{\cdot}_{i}\in\mathcal{C}\sep c_{i}>0\sep\forall i\in{\cal I};\label{eq:RSR}\\
 & \qquad\qquad c_{i}\smlnorm w\leq\smldblangle w_{i}\sep\forall w\in\mathscr{W}_{i}\sep\forall i\in{\cal I};\text{ and }\nonumber \\
 & \qquad\qquad\smldblangle{A[i]w}_{j}\leq\gamma\smldblangle w_{i}\sep\forall w\in\mathscr{W}_{i}\text{ s.t. }A[i]w\in\mathscr{W}_{j}\sep\forall(i,j)\in{\cal J}\}.\nonumber 
\end{align}
Clearly the preceding provides an upper bound for $\gamma^{\ast}$.
Whether it provides a lower bound for $\rho_{\cjsr}(A[\cdot];{\cal J})$
depends on the choice of $\mathcal{C}$; this is the case if e.g.\ $\mathcal{C}$
is taken to be the set of norms on $\reals^{d_{w}}$, or some class
capable of approximating these norms arbitrarily well (such as homogeneous
polynomials). 

In practice, the main difficulty with all of these objects is computational,
with the exact computation of the (C)JSR known to be an NP--hard
problem. Nonetheless, significant progress has been made in calculating
approximate upper bounds for both of these quantities, with an (in
principle) arbitrarily high degree of accuracy (\citealp{PJ08LAA,LPJ19ITAC,LPJ20SJCO}),
using semidefinite (SDP) and sum of squares (SOS) programming. With
respect to approximating the RJSR, there is naturally the risk of
choosing ${\cal C}$ to be so broad a class of functions that the
approximation of $\rho_{\rsr,{\cal C}}$ becomes infeasible. As explained
in \appref{computation}, building on the approach of \citet{FTCMM02Auto},
we therefore take as our starting point the SDP and SOS programs used
to approximate the (C)JSR, relaxing the analogues of \eqref{cjsr-norms}
in the direction of \eqref{RSR}, for the case where the $\{\mathscr{W}_{i}\}_{i\in{\cal I}}$
are convex cones (as is appropriate for a CKSVAR). This entails relating
${\cal C}$ to either a certain class of quadratic functions (for
the SDP program) or homogeneous polynomials of a given (even) order
(for the SOS program), and has the additional benefit that our approximate
calculations of $\rho_{\rsr,{\cal C}}$ are guaranteed to provide
a lower bound on the estimate of $\rho_{\cjsr}$ yielded by these
programs, and thus a less conservative upper bound $\gamma^{\ast}$
in practice.\footnote{Because the SDP and SOS algorithms that we use to compute an approximate
upper bound for the JSR correspond to restricted versions (i.e.\ with
fewer free parameters) of the corresponding algorithms used to compute
an upper bound for the CJSR, it is also the case that our estimate
of $\rho_{\cjsr}$ lower bounds our estimate of $\rho_{\jsr}$.} All of these quantities are calculated by our R package, \texttt{thresholdr}.

\subsection{In the CKSVAR}

\label{subsec:stable-cksvar}

In the present setting, the role of \eqref{detsystem} is played by
\eqref{determ}, the deterministic system associated to a canonical
CKSVAR, rendered here in companion form (and suppressing the `hat'
decorations) as 
\begin{equation}
\b z_{t}=\b F(\b y_{t-1})\b z_{t-1}=\b F[\sigma(\b z_{t-1})]\b z_{t-1}=\sum_{i\in{\cal I}}\b F[i]\indic\{\b z_{t-1}\in\mathscr{Z}_{i}\}\b z_{t-1},\label{eq:detCKSVAR}
\end{equation}
where $\b y_{t-1}=(y_{t-1},\ldots,y_{t-k})^{\trans}$, and
\[
\b F(\b y_{t-1})\defeq\begin{bmatrix}\Phi_{1}(y_{t-1}) & \cdots & \Phi_{k-1}(y_{t-k+1}) & \Phi_{k}(y_{t-k})\\
I_{p}\\
 & \ddots\\
 &  & I_{p}
\end{bmatrix}.
\]
Here there are (at most) $\smlabs{{\cal I}}=2^{k}$ distinct states.
If we associate with each $i\in{\cal I}$ a distinct $s(i)\in\{0,+1\}^{k}$,
we can take 
\[
\mathscr{Z}_{i}\defeq\{(y_{t-1},x_{t-1}^{\trans},\ldots,y_{t-k},x_{t-k}^{\trans})^{\trans}\in\reals^{kp}\mid\indic^{+}(y_{t-j})=s_{j}(i)\sep\forall j\in\{1,\ldots,k\}\},
\]
$\sigma(\b z_{t-1})=i$ if $\b z_{t-1}\in\mathscr{Z}_{i}$, and $\b F[i]$
to be the autoregressive matrix that applies in this case.

Letting ${\cal F}\defeq\{\b F[i]\}_{i\in{\cal I}}$, the JSR of ${\cal F}$
provides a first estimate of the stability degree $\gamma^{\ast}$
of the system. However, since the sign pattern of the first $k-1$
elements of $\b y_{t-1}$ must propagate forwards to $\b y_{t}$,
any given state $i\in{\cal I}$ can only be succeeded by two other
states (which differ only according to the sign of $y_{t}$). Thus
we can generally obtain a tighter estimate via the CJSR, which reduces
the set of possible transitions from the $\smlabs{{\cal I}}^{2}=2^{2k}$
implicitly permitted by the calculation of the JSR, to $2\smlabs{{\cal I}}=2^{k+1}$.
An even lower estimate of $\gamma^{\ast}$ -- which may prove decisive
for establishing the stability of \eqref{detCKSVAR} -- can be obtained
via the RJSR, at the cost of greater computation time. \thmref{ergodicity}
thus has a noteworthy advantage over approaches that utilise bounds
on the (C)JSR to directly establish the ergodicity of the corresponding
stochastic system, because the evolution of the deterministic system
is much more tightly constrained, particularly when the innovations
$\{u_{t}\}$ have full support, as assumed in \assref{err}.

For especially tractable instances of the CKSVAR, additional sufficient
conditions for the stability of \eqref{detCKSVAR} may also be available.
For example, a result of \citet[Suppl., Lem.~1]{DJH11} implies that
for the censored dynamic Tobit (\eqref{censored-Tobit} above), a
sufficient condition for stability is that the `censored' autoregressive
polynomial $1-\sum_{i=1}^{k}[\phi_{i}^{+}]_{+}L^{i}$ should have
all its roots outside the unit circle. (Interestingly, as their result
indicates, control over the roots of the various autoregressive regimes
is not even \emph{necessary} to ensure the stability of a CKSVAR.)

\saveexamplex{}

\exname{\ref*{exa:monetary}}
\begin{example}[monetary policy; ctd]
 From the canonical representation \eqref{example-canon} of the
model, we have
\begin{align*}
\tilde{\Phi}_{1}^{+} & =\begin{bmatrix}\psi & \gamma(\chi\kappa_{1}-\psi)\kappa_{1}\\
0 & \chi\kappa_{1}
\end{bmatrix} & \tilde{\Phi}_{1}^{-} & =\begin{bmatrix}\psi-\chi\tau_{\mu}\kappa_{\mu} & \gamma(\chi\kappa_{1}-\psi)\kappa_{1}\\
-\chi\theta(1-\mu)\kappa_{\mu} & \chi\kappa_{1}
\end{bmatrix}.
\end{align*}
Ordinarily, the stability of this nonlinear system could not be assessed
simply by a consideration of the eigenvalues of these matrices; but
in the present setting it may be verified that these are simultaneously
(lower) triangularised by $T\defeq[\begin{smallmatrix}1 & -\gamma\kappa_{1}\\
0 & 1
\end{smallmatrix}]$, with in particular
\begin{align*}
T\tilde{\Phi}_{1}^{+}T^{-1} & =\begin{bmatrix}\psi & 0\\
0 & \chi\kappa_{1}
\end{bmatrix}, & T\tilde{\Phi}_{1}^{-}T^{-1} & =\begin{bmatrix}\psi & 0\\
-\chi\theta(1-\mu)\kappa_{\mu} & \chi\kappa_{\mu}
\end{bmatrix},
\end{align*}
and hence 
\begin{equation}
\rho_{\jsr}=\max\{\smlabs{\psi},\chi\kappa_{1},\chi\kappa_{\mu}\}=\max\{\smlabs{\psi},\chi\kappa_{\mu}\}.\label{eq:jsr-exact}
\end{equation}
Thus $(\tilde{\imath}_{t},\tilde{\pi}_{t})$, and therefore $(i_{t},\pi_{t})$,
is geometrically ergodic if the preceding is strictly less than unity.
\end{example}
\restoreexamplex{}

\subsection{Numerical examples}

\label{subsec:numerical-examples}

To illustrate the practical utility of the stability criteria developed
above, we compute upper bounds for the JSR, CJSR and RJSR (using our
\R{} package, \texttt{thresholdr}) for some of our running examples.
These upper bounds are denoted with a `bar' (as $\bar{\rho}_{\jsr}$,
etc.)\ to distinguish them from the actual quantities, which in general
cannot be computed exactly.

\begin{table}
\begin{centering}
\begin{tabular}{r@{\extracolsep{0pt}.}lr@{\extracolsep{0pt}.}lr@{\extracolsep{0pt}.}lr@{\extracolsep{0pt}.}lr@{\extracolsep{0pt}.}lr@{\extracolsep{0pt}.}lr@{\extracolsep{0pt}.}lr@{\extracolsep{0pt}.}lr@{\extracolsep{0pt}.}lr@{\extracolsep{0pt}.}lr@{\extracolsep{0pt}.}lr@{\extracolsep{0pt}.}lr@{\extracolsep{0pt}.}lr@{\extracolsep{0pt}.}l}
\toprule 
\addlinespace
\multicolumn{2}{c}{$\chi$} & \multicolumn{2}{c}{$\theta$} & \multicolumn{2}{c}{$\psi$} & \multicolumn{2}{c}{$\bar{\rho}_{\jsr}$} & \multicolumn{2}{c}{} & \multicolumn{2}{c}{$\chi$} & \multicolumn{2}{c}{$\theta$} & \multicolumn{2}{c}{$\psi$} & \multicolumn{2}{c}{$\bar{\rho}_{\jsr}$} & \multicolumn{2}{c}{} & \multicolumn{2}{c}{$\chi$} & \multicolumn{2}{c}{$\theta$} & \multicolumn{2}{c}{$\psi$} & \multicolumn{2}{c}{$\bar{\rho}_{\jsr}$}\tabularnewline\addlinespace
\midrule
\addlinespace
0&2 & $-$0&5 & 0&1 & 0&145 & \multicolumn{2}{c}{} & 0&7 & $-$1&0 & 0&1 & 0&4 & \multicolumn{2}{c}{} & 0&99 & $-$0&5 & 0&1 & 0&72\tabularnewline
0&2 & $-$0&5 & 0&5 & 0&5 & \multicolumn{2}{c}{} & 0&7 & $-$1&0 & 0&5 & 0&5 & \multicolumn{2}{c}{} & 0&99 & $-$0&5 & 0&5 & 0&72\tabularnewline
0&2 & $-$0&5 & 0&9 & 0&9 & \multicolumn{2}{c}{} & 0&7 & $-$1&0 & 0&9 & 0&9 & \multicolumn{2}{c}{} & 0&99 & $-$0&5 & 0&9 & 0&9\tabularnewline
\bottomrule
\end{tabular}
\par\end{centering}
\caption{$\bar{\rho}_{\protect\jsr}$ for \exaref{monetary}}

\label{tbl:monetary}
\end{table}

\begin{table}
\begin{centering}
\begin{tabular}{r@{\extracolsep{0pt}.}lr@{\extracolsep{0pt}.}lr@{\extracolsep{0pt}.}lr@{\extracolsep{0pt}.}lccc}
\toprule 
\addlinespace
\multicolumn{2}{c}{$\phi_{1}^{+}$} & \multicolumn{2}{c}{$\phi_{1}^{-}$} & \multicolumn{2}{c}{$\phi_{2}^{+}$} & \multicolumn{2}{c}{$\phi_{2}^{-}$} & \multicolumn{1}{c}{$\bar{\rho}_{\jsr}$} & \multicolumn{1}{c}{$\bar{\rho}_{\cjsr}$} & \multicolumn{1}{c}{$\bar{\rho}_{\rsr}$}\tabularnewline\addlinespace
\midrule
\addlinespace
0&6 & 0&2 & 0&3 & 0&1 & $0.925$ & $0.925$ & $0.925$\tabularnewline
0&6 & 0&3 & 0&4 & 0&1 & $1.000$ & $1.000$ & $1.000$\tabularnewline
0&7 & 0&2 & $-$0&1 & 0&0 & $0.700$ & $0.500$ & $0.500$\tabularnewline
1&2 & 0&6 & $-$1&2 & $-$0&6 & $1.245$ & $1.118$ & $1.095$\tabularnewline
1&0 & 0&5 & $-$0&97 & $-$0&5 & $1.105$ & $1.001$ & $0.985$\tabularnewline\addlinespace
\bottomrule
\end{tabular}
\par\end{centering}
\caption{$\bar{\rho}_{\protect\jsr}$, $\bar{\rho}_{\protect\cjsr}$ and $\bar{\rho}_{\protect\rsr}$
for \exaref{univariate}}

\label{tbl:univariate}
\end{table}

\saveexamplex{}

\exname{\ref*{exa:monetary}}
\begin{example}[monetary policy; ctd]
 In the model \eqref{cksvar-natrate}, there are effectively only
two autoregressive regimes once the model is rendered in canonical
form, as per \eqref{example-canon}. Since there are no restrictions
on the transitions between these regimes, the JSR and CJSR coincide.
\tblref{monetary} therefore reports only estimated upper bounds for
the JSR for selected parametrisations of the model as $(\chi,\theta,\psi)$
vary, holding $(\mu,\gamma)=(0.5,1.5)$ fixed. It may be verified
that these numerical upper bounds in fact exactly equal the JSR for
this model, computed as per \eqref{jsr-exact} above.
\end{example}
\begin{figure}[!p]
\begin{centering}
\includegraphics[width=1\textwidth]{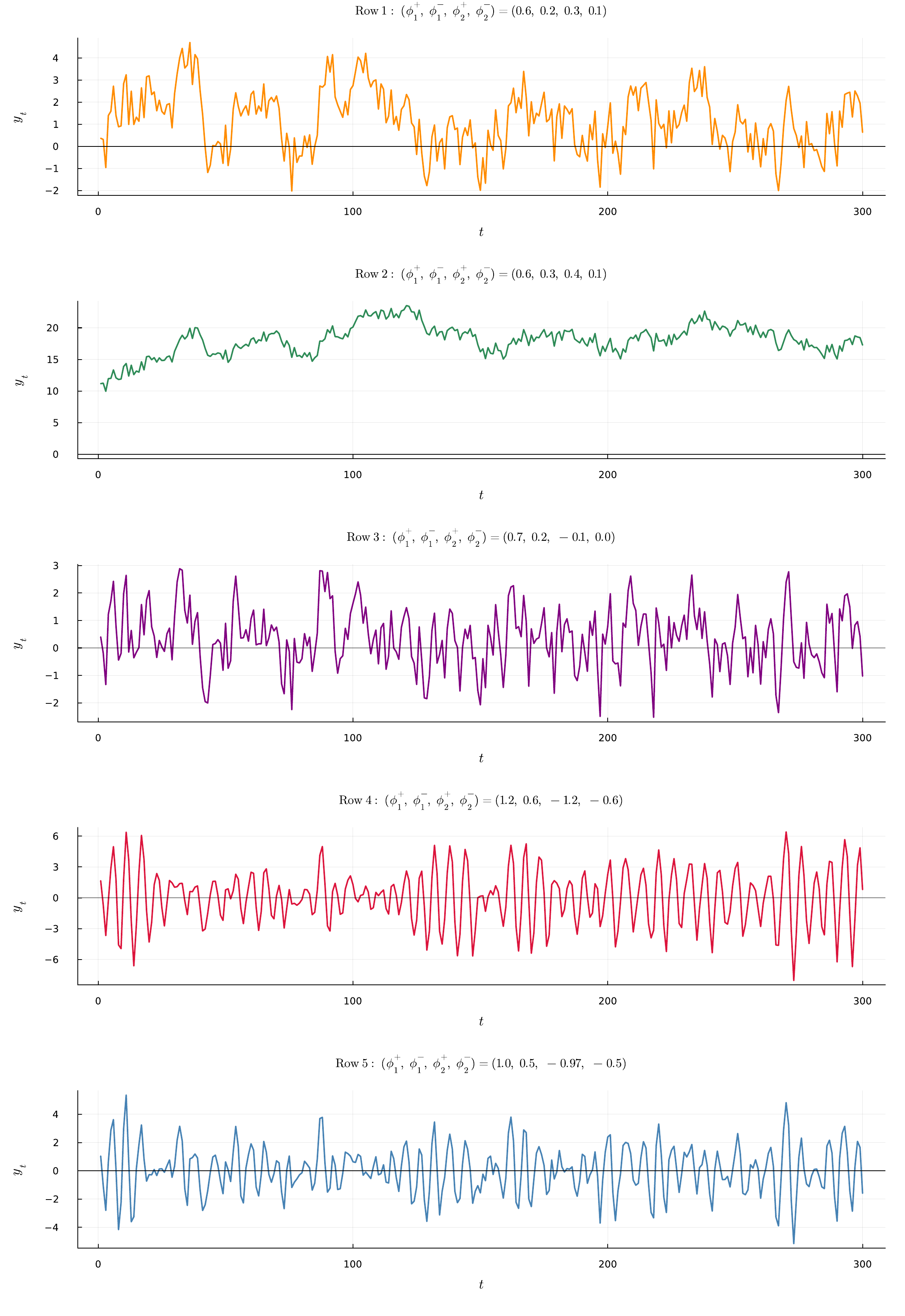}
\par\end{centering}
\caption{\protect\label{fig:trajectories}Simulated trajectories for the models
in \tblref{univariate}}
\end{figure}

\restoreexamplex{}

\saveexamplex{}

\exname{\ref*{exa:univariate}}
\begin{example}[univariate; ctd]
 We report upper bounds on the JSR, CJSR and RJSR for some parametrisations
of the univariate model \eqref{univariate-case} with two lags ($k=2$),
allowing $\phi_{i}^{-}$ to take nonzero values (unlike in the censored
dynamic Tobit model). In the presence of more than one lag, the regime-switching
behaviour is sufficiently rich to generate meaningful differences
between our estimated upper bounds for these three quantities. The
final two cases in \tblref{univariate} indicate that the (estimated)
RJSR can be a significantly less conservative criterion for establishing
stability than is the CJSR or the JSR. Indeed, in the final case,
the estimated upper bound for the RJSR is strictly below $1$ while
those for both the CJSR and JSR exceed $1$. With the exception of
the fourth process in the table, the stability or instability suggested
by our RJSR estimate is consistent with the evidence provided by the
simulated time plots in \figref{trajectories}. The fact that the
RJSR is, in general, only a conservative criterion for the stability
of these processes is illustrated by that fourth process, which is
apparently stable even though our estimate of (an upper bound for)
the RJSR exceeds unity.

A graphical illustration of our stability conditions may also be provided
for the special case of the first-order ($k=1$) univariate model
with $c=0$,
\begin{equation}
y_{t}=\phi^{+}y_{t-1}^{+}+\phi^{-}y_{t-1}^{-}+u_{t},\label{eq:cksvar-univariate}
\end{equation}
which is exactly the model considered in \citet{PW84JAP}. Consider
the associated deterministic system,
\begin{equation}
y_{t}=\begin{cases}
\phi^{+}y_{t-1}, & \text{if }y_{t-1}\geq0,\\
\phi^{-}y_{t-1}, & \text{if }y_{t-1}<0.
\end{cases}\label{eq:scalar-lagged}
\end{equation}
The sign pattern of $(\phi^{+},\phi^{-})$ determines the sequence
of regimes (positive or negative $y_{t}$) that is admissible by this
system. If:
\begin{enumerate}
\item both coefficients are non-negative ($\phi^{+}\geq0$ and $\phi^{-}\geq0$):
the sign of $y_{t}$ equals that of $y_{0}$ for all $t$, and thus
stability requires both to lie below one, i.e.\ $\max\{\phi^{+},\phi^{-}\}<1$;
\item the coefficients have opposite signs ($\phi^{+}\leq0$ and $\phi^{-}>0$;
or $\phi^{+}>0$ and $\phi^{-}\leq0$): after at most one period,
$y_{t}$ enters the regime associated with the positive coefficient,
and remains there for all subsequent $t$, and thus stability requires
again that $\max\{\phi^{+},\phi^{-}\}<1$; and
\item both coefficients are negative, the sign of $y_{t}$ alternates every
period, so that $y_{t+2}=\phi^{+}\phi^{-}y_{t}$, which is stable
if and only if $\phi^{+}\phi^{-}<1$.
\end{enumerate}
The implied stability region is depicted in \figref{scalar-stability}
(cf.\ Figure~2.1 in \citealp{PW84JAP}).

\begin{figure}
\begin{centering}
\includegraphics[viewport=0bp 25bp 520bp 485bp,clip,width=0.5\columnwidth]{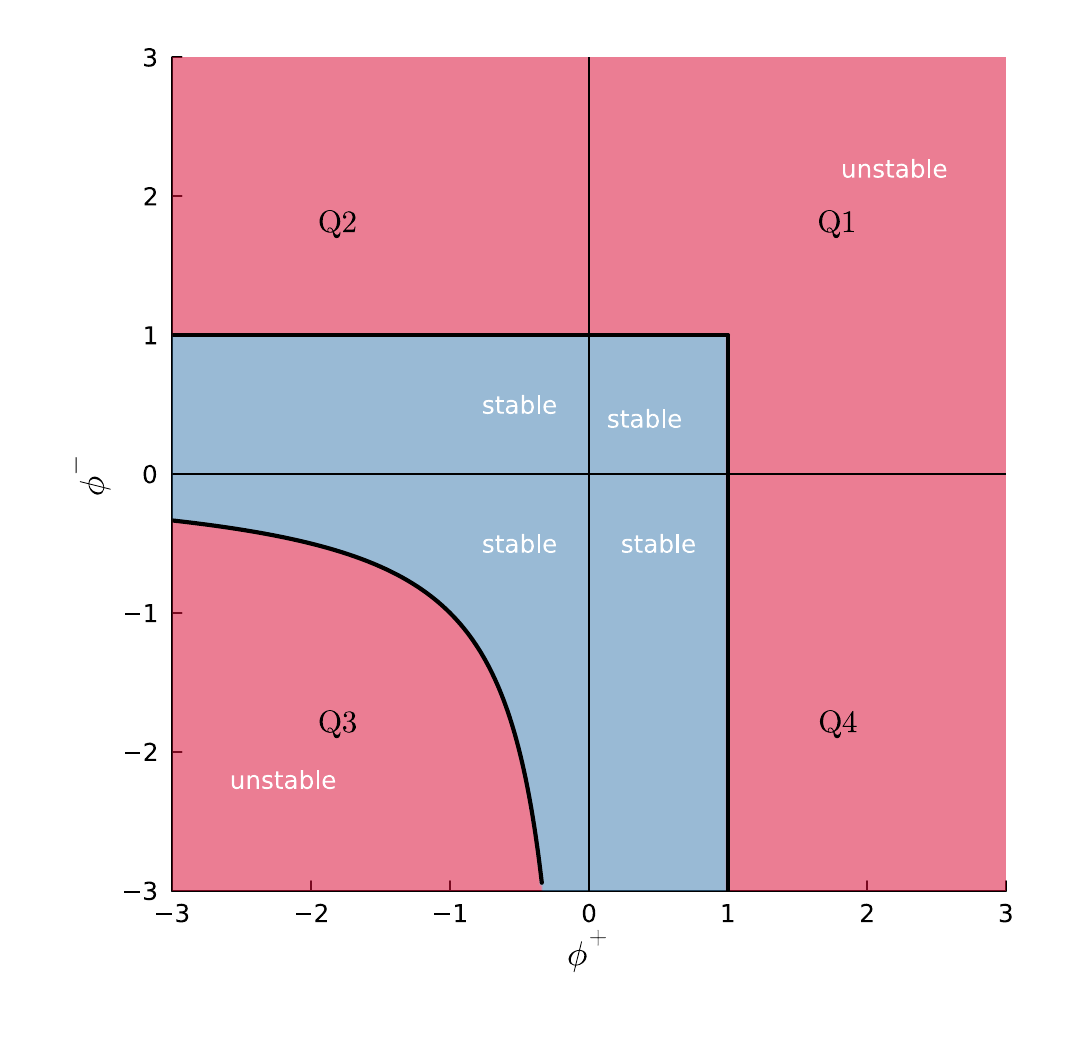}
\par\end{centering}
\caption{\protect\label{fig:scalar-stability}Stability region for the univariate
CKSVAR \eqref{cksvar-univariate}}
\end{figure}

In terms of the stability criteria developed in \subsecref{stabilityabstract},
in this model $\rho_{\jsr}=\max\{\smlabs{\phi^{+}},\smlabs{\phi^{-}}\}$,
and so $\rho_{\jsr}<1$ delineates a set that is considerably smaller
than the stable region in \figref{scalar-stability}. This is because
the JSR calculation implicitly maintains that \emph{any} possible
succession of regimes (here, sign patterns for $y_{t}$) may be generated
by the model, irrespective of the sign pattern of $(\phi^{+},\phi^{-})$.
By contrast, the calculation of the CJSR respects the sequence of
regimes permitted under \eqref{scalar-lagged}, and thus
\[
\rho_{\cjsr}=\begin{cases}
\max\{\phi^{+},\phi^{-}\} & \text{if at least one of }\phi^{+}>0\text{ or }\phi^{-}>0,\\
(\phi^{+}\phi^{-})^{1/2} & \text{otherwise;}
\end{cases}
\]
which excepting the case where both $\phi^{+}$ and $\phi^{-}$ are
positive, yields a criterion that is strictly less demanding than
that yielded by the JSR. Indeed, the parameter values for which $\rho_{\cjsr}<1$
correspond exactly to the stability region depicted in \figref{scalar-stability}.
\end{example}
\restoreexamplex{}

\saveexamplex{}

\exname{\ref*{exa:explosive}}
\begin{example}[explosive; ctd]
 Although the two autoregressive matrices both have all eigenvalues
lying within the unit circle, our estimated upper bound for the JSR
of $\{\Phi^{+},\Phi^{-}\}$ is $1.31$, consistent with the explosive
trajectories obtained when simulating the process.
\end{example}
\restoreexamplex{}

\section{Conclusion}

\label{sec:conclusion}

The CKSVAR provides a flexible yet tractable framework in which to
structurally model vector time series that are subject to an occasionally
binding constraint, and more general threshold nonlinearities. Nonetheless,
even that seemingly limited amount of nonlinearity radically changes
the properties of the model. The usual criteria for stationarity and
ergodicity, in terms of the roots of the autoregressive polynomial(s),
no longer apply, being neither necessary nor sufficient for the CKSVAR
to generate a stationary time series. We must instead employ other
criteria, such as the joint spectral radius or the related quantities
developed above, to establish stationarity via control over the deterministic
subsystem of the CKSVAR.

{\singlespacing\bibliographystyle{ecta}
\bibliography{cksvar}
}

\appendix

\section{The canonical CKSVAR}

\label{app:canonical}
\begin{proof}[Proof of \propref{canonical}]
 \textbf{(i).} We have to verify that $(\tilde{y}_{t},\tilde{x}_{t})$
as defined in \eqref{canon-vars} are generated by a canonical CKSVAR:
i.e.\ that \assref{dgp} holds for $\{(\tilde{y}_{t},\tilde{x}_{t})\}$
with $\tilde{\Phi}_{0}=\Ican[p]$. Premultiplying \eqref{var-pm}
by $Q$ and using \eqref{canon-vars}--\eqref{canon-polys} immediately
yields \eqref{tildeVAR}, which has the same form as \eqref{var-pm}.
Further, the first two equations of \eqref{canon-vars} yield $\tilde{y}_{t}^{+}=\bar{\phi}_{0,yy}^{+}y_{t}^{+}$
and $\tilde{y}_{t}^{-}=\bar{\phi}_{0,yy}^{-}y_{t}^{-}$, where $\bar{\phi}_{0,yy}^{\pm}>0$
under \assref{dgp}\enuref{dgp:wlog}, whence $\tilde{y}_{t}^{+}\geq0$
and $\tilde{y}_{t}^{-}\leq0$, with $\tilde{y}_{t}^{+}\cdot\tilde{y}_{t}^{-}=0$
in every period. Thus defining
\begin{equation}
\tilde{y}_{t}\defeq\tilde{y}_{t}^{+}+\tilde{y}_{t}^{-}=\bar{\phi}_{0,yy}^{+}y_{t}^{+}+\bar{\phi}_{0,yy}^{-}y_{t}^{-},\label{eq:tildeydef}
\end{equation}
it follows that $\tilde{y}_{t}^{+}=\max\{\tilde{y}_{t},0\}$ and $\tilde{y}_{t}^{-}=\min\{\tilde{y}_{t},0\}$;
and hence $\tilde{y}_{t}^{\pm}$ have the same form as in \eqref{y-threshold}
(with $b=0$). Deduce \assref{dgp}\enuref{dgp:defn} holds for $\{(\tilde{y}_{t},\tilde{x}_{t})\}$.
We next note that
\[
Q\Phi_{0}=\begin{bmatrix}1 & -\phi_{0,yx}^{\trans}\Phi_{0,xx}^{-1}\\
0 & I
\end{bmatrix}\begin{bmatrix}\phi_{0,yy}^{+} & \phi_{0,yy}^{-} & \phi_{0,yx}^{\trans}\\
\phi_{0,xy}^{+} & \phi_{0,xy}^{-} & \Phi_{0,xx}
\end{bmatrix}=\begin{bmatrix}\bar{\phi}_{0,yy}^{+} & \bar{\phi}_{0,yy}^{-} & 0\\
\phi_{0,xy}^{+} & \phi_{0,xy}^{-} & \Phi_{0,xx}
\end{bmatrix}
\]
where $\Phi_{0,xx}$ is invertible; and
\[
\Ican[p]P^{-1}=\begin{bmatrix}1 & 1 & 0\\
0 & 0 & I_{p-1}
\end{bmatrix}\begin{bmatrix}\bar{\phi}_{0,yy}^{+} & 0 & 0\\
0 & \bar{\phi}_{0,yy}^{-} & 0\\
\phi_{0,xy}^{+} & \phi_{0,xy}^{-} & \Phi_{0,xx}
\end{bmatrix}=\begin{bmatrix}\bar{\phi}_{0,yy}^{+} & \bar{\phi}_{0,yy}^{-} & 0\\
\phi_{0,xy}^{+} & \phi_{0,xy}^{-} & \Phi_{0,xx}
\end{bmatrix}
\]
whence $\tilde{\Phi}_{0}=Q\Phi_{0}P=\Ican[p]$, as required for a
canonical CKSVAR; the remaining parts of \assref{dgp} for $\{(\tilde{y}_{t},\tilde{x}_{t})\}$
follow immediately.

\textbf{(ii).} By dividing \eqref{var-pm} through by $\bar{\phi}_{0,yy}^{+}>0$,
we can always obtain a representation of the CKSVAR in which $\bar{\phi}_{0,yy}^{+}=1$.
Because $y_{t}^{-}$ is not observed, it may be rescaled such that
$\bar{\phi}_{0,yy}^{-}=1$ also, without affecting the distribution
of the observed series $\{(y_{t}^{+},x_{t})\}$. Now apply the argument
from part~(i), and note in particular that \eqref{tildeydef} now
becomes
\[
\tilde{y}_{t}=\bar{\phi}_{0,yy}^{+}y_{t}^{+}+\bar{\phi}_{0,yy}^{-}y_{t}^{-}=y_{t}^{+}+y_{t}^{-}=y_{t}.\qedhere
\]
\end{proof}

\section{Ergodicity of the CKSVAR}

\label{app:ergodic}

We first state a well-known auxiliary result on the stability of continuous
and homogeneous of degree one (HD1) systems (whose proof is provided
for the sake of completeness). For $w\in\reals^{d_{w}}$ and $\Lambda:\reals^{d_{w}}\setmap\reals^{d_{w}}$,
let $\Lambda^{k}(w)\defeq\Lambda[\Lambda^{k-1}(w)]$ with $\Lambda^{0}(w)\defeq w$.
\begin{lem}
\label{lem:lyapunov}Suppose $w_{t}=\Lambda(w_{t-1})$, where $\Lambda:\reals^{d_{w}}\setmap\reals^{d_{w}}$
is continuous and homogeneous of degree one (HD1), and that $\Lambda^{k}(w_{0})\goesto0$
for all $w_{0}$ in an open neighbourhood of $0$. Then there exist
a $\gamma\in(0,1)$ and a Lyapunov function $V:\reals^{d_{w}}\setmap\reals_{+}$
for which $V[\Lambda(w)]\leq\gamma V(w)$ and $\smlnorm w\leq V(w)\leq C\smlnorm w$
for all $w\in\reals^{d_{w}}$, and which is continuous and HD1.
\end{lem}
\begin{proof}
Let $S$ denote the unit sphere centred at zero, and fix $r\in(0,1)$.
For each $k\in\naturals$, let
\[
W_{k}\defeq\{w\in S\mid\smlnorm{\Lambda^{k}(w)}<r\}.
\]
By continuity of $\Lambda$, $W_{k}$ is open (relative to $S$).
Since $\Lambda$ is HD1, it follows from the maintained assumptions
that $\Lambda^{k}(w)\goesto0$ for every $w\in S$. Hence for every
$w\in S$ there exists some $k\in\naturals$ such that $w\in W_{k}$.
Thus $\{W_{k}\}_{k\in\naturals}$ is an open cover for $S$; by compactness,
there exists a $K<\infty$ such that $S=\Union_{k=1}^{K}W_{k}$. Hence
for every $w\in S$, there exists a $k\in\{1,\ldots,K\}$ such that
$\smlnorm{\Lambda^{k}(w)}\leq r$. Since $\Lambda$ is HD1, it follows
that for every $w\in\reals^{d_{w}}$, there exists a $k\in\{1,\ldots,K\}$
such that
\[
\smlnorm{\Lambda^{k}(w)}\leq r\smlnorm w.
\]

Suppose that $w_{0}\in S$. By the preceding, there will be a set
of times $\mathcal{T}\defeq\{t_{i}\}_{i=0}^{\infty}$ (depending on
$w_{0}$), with $t_{0}=0$ and $1\leq t_{i}-t_{i-1}\leq K$, such
that 
\begin{equation}
\smlnorm{w_{t_{i}}}\leq r\smlnorm{w_{t_{i-1}}}\sep\forall i\in\naturals.\label{eq:contraction}
\end{equation}
For $T\in\naturals$, we may choose an $i(T)\in\naturals\union\{0\}$
such that $t_{i(T)}\in\mathcal{T}$ satisfies $0\leq T-t_{i(T)}<K$,
and hence $i(T)\geq\smlfloor{T/K}$. Iterating the inequalities \eqref{contraction}
yields, for $T\geq K$,
\[
\smlnorm{w_{t_{i(T)}}}\leq r\smlnorm{w_{t_{i(T)-1}}}\leq\cdots\leq r^{i(T)}\smlnorm{w_{t_{0}}}\leq r^{\smlfloor{T/K}}\smlnorm{w_{0}}=r^{\smlfloor{T/K}}.
\]
In addition, there exists a $C_{0}<\infty$ such that
\[
\smlnorm{w_{T}}=\smlnorm{\Lambda^{T-t_{i(T)}}(w_{t_{i(T)}})}\leq\max\left\{ 1,\sup_{w\in S}\smlnorm{\Lambda(w)}\right\} ^{K}\smlnorm{w_{t_{i(T)}}}\leq C_{0}\smlnorm{w_{t_{i(T)}}}
\]
where we have again used that $\Lambda$ is HD1. It follows that for
all $T$ sufficiently large, there exists a $\gamma\in(0,1)$ such
that
\[
\smlnorm{\Lambda^{T}(w_{0})}=\smlnorm{w_{T}}\leq C_{0}r^{\smlfloor{T/K}}\leq\gamma^{T},
\]
for all $w_{0}\in S$. Thus we may choose $M\in\naturals$ such that
\begin{equation}
\smlnorm{\Lambda^{M}(w)}\leq\gamma^{M}\smlnorm w\label{eq:PsiM}
\end{equation}
for all $w\in\reals^{d_{w}}$.

Now we define $V(w)\defeq\sum_{k=0}^{M-1}\gamma^{-k}\smlnorm{\Lambda^{k}(w)}$
(cf.\ the proof of Theorem~2 in \citealp{Lieb2005}); clearly it
is continuous and HD1, and $V(w)\ge\smlnorm w$ (recall $\Lambda^{0}(w)\defeq w$).
Moreover
\begin{align*}
V[\Lambda(w)]=\gamma\sum_{k=1}^{M}\gamma^{-k}\smlnorm{\Lambda^{k}(w)} & =\gamma\sum_{k=1}^{M-1}\gamma^{-k}\smlnorm{\Lambda^{k}(w)}+\gamma^{-M+1}\smlnorm{\Lambda^{M}(w)}\\
 & \leq\gamma\sum_{k=1}^{M-1}\gamma^{-k}\smlnorm{\Lambda^{k}(w)}+\gamma\smlnorm w=\gamma V(w),
\end{align*}
for $w\in\reals^{d_{w}}$, where the inequality follows from \eqref{PsiM}.
Finally, the existence of a $C<\infty$ such that $V(w)\leq C\smlnorm w$
follows by taking $C\defeq\sup_{w\in S}\sum_{k=0}^{M-1}\gamma^{-k}\smlnorm{\Lambda^{k}(w)}<\infty$.
\end{proof}

We next provide a lemma that establishes the ergodicity of a broader
class of autoregressive processes, from which \thmref{ergodicity}
will then follow as a special case. This result is closely related
to those of \citet[Thms~4.2 and 4.5]{CT85AAP} and \citet[Thm.~3.2]{CP99StSin}.
Recall from \subsecref{ergodicstable} that a deterministic process
is said to be \emph{stable} if it converges to zero, irrespective
of its initialisation.
\begin{lem}
\label{lem:ergodic} Let $\{w_{t}\}_{t\in\naturals}$ be a random
sequence in $\reals^{p}$, and $\b w_{t-1}\defeq(w_{t-1}^{\trans},\ldots,w_{t-k}^{\trans})^{\trans}$.
Suppose
\begin{equation}
w_{t}=\mu(\b w_{t-1})+\b A(\b w_{t-1})\b w_{t-1}+u_{t}\label{eq:wsystem}
\end{equation}
where $\b A:\reals^{kp}\setmap\reals^{p\times kp}$ is such that $\b w\elmap\b A(\b w)\b w$
is continuous and HD1, $\mu(\b w)$ is locally bounded with $\mu(\b w)=o(\smlnorm{\b w})$
as $\smlnorm{\b w}\goesto\infty$, and that $\{u_{t}\}$ satisfies
\assref{err} with $\b w_{t-1}$ in place of $\b z_{t-1}$. If the
associated deterministic system
\[
\hat{w}_{t}=\b A(\hat{\b w}_{t-1})\hat{\b w}_{t-1},
\]
is stable, then $\{\b w_{t}\}_{t\in\naturals}$ is $\mathcal{Q}$-geometrically
ergodic, for $\mathcal{Q}(\b w)\defeq1+\smlnorm{\b w}^{m_{0}}$, with
a stationary distribution that is equivalent to Lebesgue measure on
$\reals^{kp}$, and which has finite $m_{0}$th moment. If in addition
$\expect\smlnorm{\b w_{0}}^{m_{0}}<\infty$, then $\sup_{t\in\naturals}\expect\smlnorm{\b w_{t}}^{m_{0}}<\infty$.
\end{lem}

\begin{proof}
We shall verify that $\{\b w_{t}\}_{t\in\naturals}$ satisfies conditions~(i)--(iii)
of Proposition~1 in \citet{Lieb2005}. Under \assref{err} (with
$\b w_{t-1}$ in place of $\b z_{t-1}$) we have
\begin{equation}
w_{t}=\mu(\b w_{t-1})+\b A(\b w_{t-1})\b w_{t-1}+\Sigma(\b w_{t-1})\err_{t}.\label{eq:wepsilon}
\end{equation}
Now let $f$ denote the density of $\err_{t}$, $g$ the density of
$w_{t}$ given $\b w_{t-1}$, and $\varpi_{t}\in\reals^{p}$ and $\b{\varpi}_{t-1}\in\reals^{kp}$
values that may be taken by $w_{t}$ and $\b w_{t-1}$ respectively.
Noting
\begin{equation}
\err_{t}=\Sigma(\b w_{t-1})^{-1}[w_{t}-\mu(\b w_{t-1})-\b A(\b w_{t-1})\b w_{t-1}],\label{eq:uzeta}
\end{equation}
we have by the change of variables formula that
\[
g(\varpi_{t}\mid\b{\varpi}_{t-1})=\smlabs{\det\Sigma(\b{\varpi}_{t-1})}^{-1}f\{\Sigma(\b{\varpi}_{t-1})^{-1}[\varpi_{t}-\mu(\b{\varpi}_{t-1})-\b A(\b{\varpi}_{t-1})\b{\varpi}_{t-1}]\},
\]
and by the Markov property that the density $h$ of $\b w_{t-1+k}$
given $\b w_{t-1}$ is given by
\begin{equation}
h(\b{\varpi}_{t-1+k}\mid\b{\varpi}_{t-1})=\prod_{i=1}^{k}g(\varpi_{t-1+i}\mid\b{\varpi}_{t-2+i}).\label{eq:kstepdens}
\end{equation}

\textbf{(i).} Suppose $N$ is a Borel subset of $\reals^{kp}$ with
zero Lebesgue measure. By \eqref{kstepdens}, conditional on $\b w_{t-1}=\b{\varpi}_{t-1}$,
the distribution of $\b w_{t-1+k}$ has a density with respect to
Lebesgue measure. Hence $\Prob\{\b w_{t-1+k}\in N\mid\b w_{t-1}=\b{\varpi}_{t-1}\}=0$,
for all $\b{\varpi}_{t-1}\in\reals^{kp}$. 

\textbf{(ii).} Let $K\subset\reals^{kp}$ be compact, and $B$ a Borel
subset of $\reals^{kp}$ with Lebesgue measure $\leb(B)>0$. By sigma-finiteness,
we may take $B$ to be bounded without loss of generality. Then there
exist bounded sets $\{B_{i}\}_{i=1}^{k}$ such that $\b w_{t-1+k}\in B$
only if $w_{t-1+i}\in B_{i}$ for each $i\in\{1,\ldots,k\}$. Now
as $(\varpi_{t},\b{\varpi}_{t-1})$ ranges over $B_{1}\times K$,
in view of \assref{err}\enuref{err:leb} and \assref{err}\enuref{err:eigs}
the r.h.s.\ of \eqref{uzeta} remains bounded, and so $g(\varpi_{t}\mid\b{\varpi}_{t-1})$
is bounded away from zero on $B_{1}\times K$. By the same argument,
$g(\varpi_{t-1+i}\mid\b{\varpi}_{t-2+i})$ is bounded away from zero
as $(\varpi_{t-1+i},\ldots\varpi_{t},\b{\varpi}_{t-1})$ ranges over
$B_{i}\times\cdots\times B_{1}\times K$, for $i\geq1$. Hence there
are $\{\delta_{i}\}_{i=1}^{k}$ strictly positive such that
\[
\inf_{\substack{(\b{\varpi}_{t-1+k},\b{\varpi}_{t-1})\\
\in B\times K
}
}h(\b{\varpi}_{t-1+k}\mid\b{\varpi}_{t-1})\geq\prod_{i=1}^{k}\inf_{\substack{(\varpi_{t-1+i},\ldots\varpi_{t},\b{\varpi}_{t-1})\\
\in B_{i}\times\cdots\times B_{1}\times K
}
}g(\varpi_{t-1+i}\mid\b{\varpi}_{t-2+i})=\prod_{i=1}^{k}\delta_{i}\eqdef\delta>0.
\]
Deduce that
\[
\inf_{\b{\varpi}_{t-1}\in K}\Prob\{\b w_{t-1+k}\in B\mid\b w_{t-1}=\b{\varpi}_{t-1}\}\geq\delta\leb(B)>0.
\]

\textbf{(iii).} Write \eqref{wepsilon} in companion form as
\[
\b w_{t}=\b{\mu}(\b w_{t-1})+\b F(\b w_{t-1})\b w_{t-1}+\b{\Sigma}(\b w_{t-1})\b{\err}_{t}
\]
where $\b{\Sigma}(\b w_{t-1})=\diag\{\Sigma(\b w_{t-1}),0_{p(k-1)\times p(k-1)}\}$
and
\begin{align*}
\b F(\b w_{t-1}) & \defeq\begin{bmatrix}\b A(\b w_{t-1})\\
\Xi
\end{bmatrix} & \b{\mu}(\b w_{t-1}) & \defeq\begin{bmatrix}\mu(\b w_{t-1})\\
0_{p(k-1)}
\end{bmatrix} & \b{\err}_{t} & \defeq\begin{bmatrix}\varepsilon_{t}\\
0_{p(k-1)}
\end{bmatrix}
\end{align*}
for $\Xi\defeq[I_{p(k-1)},0_{p(k-1)\times p}]$. Then by the maintained
assumptions $\Lambda(\b w)\defeq\b F(\b w)\b w$ satisfies the requirements
of \lemref{lyapunov}; let $V$ denote the associated Lyapunov function.
Setting $\mathcal{Q}(\b w)\defeq1+V(\b w)^{m_{0}}$, we have by the
independence of $\b{\err}_{t}$ from $\b w_{t-1}$ that
\begin{align*}
\expect[\mathcal{Q}(\b w_{t})\mid\b w_{t-1}=\b{\varpi}_{t-1}] & =\expect[1+V[\b{\mu}(\b w_{t-1})+\b F(\b w_{t-1})\b w_{t-1}+\b{\Sigma}(\b w_{t-1})\b{\err}_{t}]^{m_{0}}\mid\b w_{t-1}=\b{\varpi}_{t-1}]\\
 & =\expect\{1+V[\Lambda(\b{\varpi}_{t-1})+\b{\mu}(\b{\varpi}_{t-1})+\b{\Sigma}(\b{\varpi}_{t-1})\b{\err}_{t}]^{m_{0}}\}\eqdef\expect R(\b{\varpi}_{t-1},\b{\err}_{t}).
\end{align*}
To verify (2) in \citet{Lieb2005}, we shall show that there exists
a $C_{0}>0$ and a $\gamma^{\prime}\in(0,1)$ such that for all $\smlnorm{\b{\varpi}_{t-1}}>C_{0}$,
\[
\expect[\mathcal{Q}(\b w_{t})\mid\b w_{t-1}=\b{\varpi}_{t-1}]=\expect R(\b{\varpi}_{t-1},\b{\err}_{t})\leq\gamma^{\prime}\mathcal{Q}(\b{\varpi}_{t-1}).
\]

To that end, let $\tau\in(0,1)$ and define
\begin{align*}
W_{(\geq)} & \defeq\{\b w\in\reals^{kp}\mid\smlnorm{\Lambda(\b w)}\geq\tau\smlnorm{\b w}\} & W_{(\leq)} & \defeq\{\b w\in\reals^{kp}\mid\smlnorm{\Lambda(\b w)}\leq\tau\smlnorm{\b w}\}.
\end{align*}
Suppose that $\b{\varpi}\in W_{(\leq)}$. Then since $V$ and $\Lambda$
are HD1, taking $x\defeq\smlnorm{\b{\varpi}}$ yields
\begin{align*}
R(\b{\varpi},\b{\err}) & =1+V[\Lambda(\b{\varpi})+\b{\mu}(\b{\varpi})+\b{\Sigma}(\b{\varpi})\b{\err}]^{m_{0}}\\
 & \leq1+x^{m_{0}}\sup_{\smlnorm{\b{\eta}}=1,\smlnorm{\Lambda(\b{\eta})}\leq\tau}V\{\Lambda(\b{\eta})+x^{-1}[\b{\mu}(x\b{\eta})+\b{\Sigma}(x\b{\eta})\b{\err}]\}^{m_{0}}\\
 & \leq1+x^{m_{0}}C_{1}\left[\tau^{m_{0}}+\sup_{\smlnorm{\b{\eta}}=1}\smlnorm{x^{-1}\b{\mu}(x\b{\eta})}^{m_{0}}+\sup_{\smlnorm{\b{\eta}}=1}\smlnorm{x^{-1}\b{\Sigma}(x\b{\eta})}^{m_{0}}\smlnorm{\b{\err}}^{m_{0}}\right].
\end{align*}
for some $C_{1}<\infty$, by \lemref{lyapunov}. By \assref{err}\enuref{err:sigmaterms}
and $\mu(\b w)=o(\smlnorm{\b w})$, it follows that by taking $C_{0}>0$
sufficiently large, we may ensure that for all $\b{\varpi}_{t-1}\in W_{(\leq)}$
with $\smlnorm{\b{\varpi}_{t-1}}>C_{0}$, 
\[
\expect R(\b{\varpi}_{t-1},\b{\err}_{t})\leq1+C_{2}\tau^{m_{0}}\smlnorm{\b{\varpi}_{t-1}}^{m_{0}}
\]
for some $C_{2}<\infty$, not depending on $\tau$. Since $\mathcal{Q}(\b{\varpi})\geq1+\smlnorm{\b{\varpi}}^{m_{0}}$,
there exists a $\gamma_{(\leq)}<1$ such that by choice of $\tau$
and $C_{0}$, we may ensure that the r.h.s.\ is less than $\gamma_{(\leq)}\mathcal{Q}(\b{\varpi}_{t-1})$
for all $\smlnorm{\b{\varpi}_{t-1}}>C_{0}$.

Now suppose that $\b{\varpi}\in W_{(\geq)}$. To handle this case,
we write
\[
R(\b{\varpi},\b{\err})=\{1+V[\Lambda(\b{\varpi})]^{m_{0}}\}\frac{1+V[\Lambda(\b{\varpi})+\b{\mu}(\b{\varpi})+\b{\Sigma}(\b{\varpi})\b{\err}]^{m_{0}}}{1+V[\Lambda(\b{\varpi})]^{m_{0}}}\eqdef\{1+V[\Lambda(\b{\varpi})]^{m_{0}}\}L_{\b{\err}}(\b{\varpi}).
\]
Since $V$ and $\Lambda$ are HD1, taking $x\defeq\smlnorm{\b{\varpi}}$
now yields that for all $\b{\varpi}\in W_{(\geq)}$,
\begin{align*}
L_{\b{\err}}(\b{\varpi}) & \leq\sup_{\smlnorm{\b{\eta}}=1,\smlnorm{\Lambda(\b{\eta})}\geq\tau}L_{\b{\err}}(x\b{\eta})\\
 & \leq1+\sup_{\smlnorm{\b{\eta}}=1,\smlnorm{\Lambda(\b{\eta})}\geq\tau}\frac{\smlabs{V[\Lambda(x\b{\eta})+\b{\mu}(x\b{\eta})+\b{\Sigma}(x\b{\eta})\b{\err}]^{m_{0}}-V[\Lambda(x\b{\eta})]^{m_{0}}}}{1+V[\Lambda(x\b{\eta})]^{m_{0}}}\\
 & \leq1+\sup_{\smlnorm{\b{\eta}}=1,\smlnorm{\Lambda(\b{\eta})}\geq\tau}\frac{\smlabs{V[\Lambda(\b{\eta})+x^{-1}\b{\mu}(x\b{\eta})+x^{-1}\b{\Sigma}(x\b{\eta})\b{\err}]^{m_{0}}-V[\Lambda(\b{\eta})]^{m_{0}}}}{x^{-m_{0}}+V[\Lambda(\b{\eta})]^{m_{0}}}
\end{align*}
The denominator on the r.h.s.\ is bounded away from zero, while for
the numerator we have
\begin{align*}
V[\Lambda(\b{\eta})+x^{-1}\b{\mu}(x\b{\eta})+x^{-1}\b{\Sigma}(x\b{\eta})\b{\err}]^{m_{0}}-V[\Lambda(\b{\eta})]^{m_{0}} & \goesto0
\end{align*}
as $x\goesto\infty$, uniformly over $\smlnorm{\b{\eta}}=1$, by \assref{err}\enuref{err:sigmaterms},
$\mu(\b w)=o(\smlnorm{\b w})$, and the uniform continuity of $V$
on compacta.

Deduce that $\sup_{\b{\varpi}\in W_{(\geq)},\smlnorm{\b{\varpi}}\geq x}L_{\b{\err}}(\b{\varpi})\goesto1$
as $x\goesto\infty$, for every $\b{\err}\in\reals^{kp}$. Under the
stated growth conditions on $\b{\mu}(\b{\varpi})$ and $\b{\Sigma}(\b{\varpi})$,
there exists a $C_{3}<\infty$ such that for all $\smlnorm{\b{\varpi}}$
sufficiently large,
\begin{align*}
L_{\b{\err}_{t}}(\b{\varpi}) & =\frac{1+V[\Lambda(\b{\varpi})+\b{\mu}(\b{\varpi})+\b{\Sigma}(\b{\varpi})\b{\err}_{t}]^{m_{0}}}{1+V[\Lambda(\b{\varpi})]^{m_{0}}}\\
 & \leq C_{3}\frac{1+\smlnorm{\Lambda(\b{\varpi})+\b{\mu}(\b{\varpi})}^{m_{0}}+\smlnorm{\b{\Sigma}(\b{\varpi})}^{m_{0}}\smlnorm{\b{\err}_{t}}^{m_{0}}}{1+\smlnorm{\Lambda(\b{\varpi})}^{m_{0}}}
\end{align*}
Since the r.h.s.\ is integrable (with respect to $\b{\err}_{t}$),
it follows by the dominated convergence theorem that for any given
$\delta>0$ we may take $C_{0}>0$ sufficiently large such that
\[
\expect[\mathcal{Q}(\b w_{t})\mid\b w_{t-1}=\b{\varpi}_{t-1}]=\expect R(\b{\varpi}_{t-1},\b{\err}_{t})\leq\{1+V[\Lambda(\b{\varpi}_{t-1})]^{m_{0}}\}(1+\delta),
\]
for all $\b{\varpi}_{t-1}\in W_{(\geq)}$ with $\smlnorm{\b{\varpi}_{t-1}}>C_{0}$;
and so for such $\b{\varpi}_{t-1}$,
\begin{align*}
\expect[\mathcal{Q}(\b w_{t})\mid\b w_{t-1}=\b{\varpi}_{t-1}] & \leq[1+\gamma^{m_{0}}V(\b{\varpi}_{t-1})^{m_{0}}](1+\delta)\\
 & =\gamma^{m_{0}}(1+\delta)\mathcal{Q}(\b{\varpi}_{t-1})+(1-\gamma^{m_{0}})(1+\delta).
\end{align*}
Finally, since $\mathcal{Q}(\b{\varpi})\geq1+\smlnorm{\b{\varpi}}^{m_{0}}$,
there exists a $\gamma_{(\geq)}<1$ such that by taking $C_{0}$ sufficiently
large, we may ensure that the r.h.s.\ is bounded above by $\gamma_{(\geq)}\mathcal{Q}(\b{\varpi}_{t-1})$
for all $\b{\varpi}_{t-1}\in W_{(\geq)}$ with $\smlnorm{\b{\varpi}_{t-1}}>C_{0}$.
Taking $\gamma^{\prime}=\max\{\gamma_{(\leq)},\gamma_{(\geq)}\}$
thus verifies (2) in \citet{Lieb2005}.

Regarding (3) in \citet{Lieb2005}, we note that
\begin{align*}
\expect[\mathcal{Q}(\b w_{t})\mid\b w_{t-1}=\b{\varpi}_{t-1}] & =\expect\{1+V[\Lambda(\b{\varpi}_{t-1})+\b{\mu}(\b{\varpi}_{t-1})+\b{\Sigma}(\b{\varpi}_{t-1})\b{\err}_{t}]^{m_{0}}\}\\
 & \leq C_{3}\{1+\smlnorm{\Lambda(\b{\varpi}_{t-1})+\b{\mu}(\b{\varpi}_{t-1})}^{m_{0}}+\smlnorm{\b{\Sigma}(\b{\varpi}_{t-1})}^{m_{0}}\expect\smlnorm{\b{\err}_{t}}^{m_{0}}\}
\end{align*}
is bounded uniformly over $\smlnorm{\b{\varpi}_{t-1}}\leq C_{0}$.
Finally, his (4) holds trivially, since $\mathcal{Q}(\b w)=1+V(\b w)^{m_{0}}$
is bounded away from zero, and bounded on compacta.

Applying Proposition~1 in \citet{Lieb2005} then yields all but the
final claim. Under the additional assumption that $\expect\smlnorm{\b w_{0}}^{m_{0}}<\infty$,
and treating the cases $\smlnorm{\b{\varpi}_{t-1}}>C_{0}$ and $\smlnorm{\b{\varpi}_{t-1}}\leq C_{0}$
separately, we have that
\[
\expect[\mathcal{Q}(\b w_{t})\mid\b w_{t-1}=\b{\varpi}_{t-1}]\leq\gamma^{\prime}\mathcal{Q}(\b{\varpi}_{t-1})+C
\]
for some $C<\infty$, and hence
\[
\expect\mathcal{Q}(\b w_{t})\leq\gamma^{\prime}\expect\mathcal{Q}(\b w_{t-1})+C.
\]
Since $\gamma^{\prime}<1$, a recursion yields that 
\[
1+\sup_{t\in\naturals}\expect\smlnorm{\b w_{t}}^{m_{0}}\leq\sup_{t\in\naturals}\expect\mathcal{Q}(\b w_{t})<\infty.\qedhere
\]
\end{proof}
\begin{proof}[Proof of \thmref{ergodicity}]
 We need only to recognise that under the stated hypotheses the system
\eqref{regimeswitching}, i.e.
\[
z_{t}=c+\sum_{i=1}^{k}\Phi_{i}(z_{t-i})z_{t-i}+u_{t}\eqdef c+\b{\Phi}(\b z_{t-1})\b z_{t-1}+u_{t},
\]
satisfies the requirements of \lemref{ergodic}. In particular, $z\elmap\Phi_{i}(z)z=\phi_{i}(y)y+\Phi_{i}^{x}x$
is clearly continuous and HD1, for each $i\in\{1,\ldots,k\}$.
\end{proof}

\section{Limit distribution theory}

\label{app:limit-dist}

This appendix provides the proof of \thmref{limittheory}.

\subsection{Formulation of the model likelihood}

We first note that by defining
\begin{align*}
\Psi_{i} & \defeq\begin{bmatrix}\psi_{i}^{+}, & \psi_{i}^{-}, & \Psi_{i}^{x}\end{bmatrix}, & \Psi_{i}(y) & \defeq\begin{bmatrix}\psi_{i}^{+}\indic^{+}(y)+\psi_{i}^{-}\indic^{-}(y), & \Psi_{i}^{x}\end{bmatrix},
\end{align*}
and $z_{t}^{\ast}\defeq(y_{t}^{+},y_{t}^{-},x_{t}^{\trans})^{\trans}$,
we may equivalently rewrite the reparametrised CKSVAR \eqref{reparm}
as
\begin{align*}
\orthm u_{t} & =\Psi_{0}(y_{t})z_{t}-\mu-\sum_{i=1}^{k}\Psi_{i}(y_{t-i})z_{t-i}=\Psi_{0}z_{t}^{\ast}-\mu-\sum_{i=1}^{k}\Psi_{i}z_{t-i}^{\ast}.
\end{align*}
By additionally defining
\[
\xi_{t}\defeq(-1,z_{t}^{\ast\trans},-z_{t-1}^{\ast\trans},\ldots,-z_{t-k}^{\ast\trans})^{\trans}
\]
and collecting
\[
\Psi\defeq(\mu,\Psi_{0},\Psi_{1},\ldots,\Psi_{k}),
\]
we may further write
\[
\Psi_{0}z_{t}^{\ast}-\mu-\sum_{i=1}^{k}\Psi_{i}z_{t-i}^{\ast}=\Psi\xi_{t}=(\xi_{t}^{\trans}\otimes I_{p})\vek(\Psi).
\]
Here $\vek(\Psi)$ is a vector with $d_{\psi}\defeq(k+1)p(p+1)+p$
entries. However, since $\Psi_{0}^{+}=[\psi_{0}^{+},\Psi_{0}^{x}]$
is a lower triangular matrix, $\vek(\Psi)$ must respect the restriction
that the $p(p-1)/2$ elements above the main diagonal of $\Psi_{0}^{+}$
are zero: i.e.\ those elements of $\vek(\Psi)$ must be identically
zero. All other elements of $\vek(\Psi)$ correspond to (distinct)
entries of the vector $\theta$ as defined immediately above the statement
of \thmref{limittheory}, which has only $d_{\theta}\defeq d_{\psi}-p(p-1)/2$
entries. Hence we may define a matrix $M\in\reals^{d_{\psi}\times d_{\theta}}$
such that each element of $M\theta$ is either zero, or equal to a
(distinct) element of $\theta$; it follows that $M$ has full column
rank. We thus obtain
\[
\orthm u_{t}=\Psi_{0}z_{t}^{\ast}-\mu-\sum_{i=1}^{k}\Psi_{i}z_{t-i}^{\ast}=(\xi_{t}^{\trans}\otimes I_{p})M\theta.
\]

Since the map $z\elmap\Psi_{0}(e_{1}^{\trans}z)z$ is invertible under
\assref{dgp}, and $u_{t}\distrib N[0,I_{p}]$ is independent of $\{z_{s}\}_{s\leq t-1}$,
it follows by the usual change of variables formula that conditional
on $\b z_{t-1}=(z_{t-1}^{\trans},\ldots,z_{t-k}^{\trans})^{\trans}$,
the density of $z_{t}$ is given by
\[
f_{\theta}(\bar{z}_{t}\mid\bar{z}_{t-1},\ldots,\bar{z}_{t-k})=\varphi_{I_{p}}\left(\Psi_{0}\bar{z}_{t}^{\ast}-\mu-\sum_{i=1}^{k}\Psi_{i}\bar{z}_{t-i}^{\ast}\right)\cdot\smlabs{\det\Psi_{0}(\bar{y}_{t})},
\]
where a `bar' denotes a real value (e.g.\ $\bar{z}_{t}$) that
may be taken by the random variable (e.g.\ $z_{t}$), and $\varphi_{I_{p}}$
denotes the $N[0,I_{p}]$ density. Observe that the preceding is invariant
to $\orthm$, which accordingly plays no further role in our analysis.
Since $\Psi_{0}(y)$ is lower triangular, its determinant is the product
of its diagonal elements; moreover only the first element of this
diagonal differs across the only two distinct values $\Psi_{0}^{+}$
and $\Psi_{0}^{-}$ taken by this matrix. We may without loss of generality
suppose that the elements of $\theta$ are ordered such that the diagonal
elements of these two matrices correspond to the first $p+1$ elements
of $\theta$, and thus
\begin{align*}
\log\smlabs{\det\Psi_{0}(y_{t})} & =\indic^{+}(y_{t})\log\Psi_{0,11}^{+}+\indic^{-}(y_{t})\log\Psi_{0,11}^{-}+\sum_{i=2}^{p}\log\Psi_{0,ii}^{+}\\
 & =\indic^{+}(y_{t})\log\theta_{1}+\indic^{-}(y_{t})\log\theta_{2}+\sum_{i=3}^{p+1}\log\theta_{i}\eqdef\ell_{t,1}(\theta)
\end{align*}
Defining
\[
\ell_{t,2}(\theta)\defeq-\tfrac{1}{2}\smlnorm{(\xi_{t}^{\trans}\otimes I_{p})M\theta}^{2},
\]
the $t$th (conditional) loglikelihood contribution is thus given
by
\begin{align}
\ell_{t}(\theta) & =\log f_{\theta}(z_{t}\mid\b z_{t-1})=-\tfrac{p}{2}\log2\pi+\ell_{t,1}(\theta)+\ell_{t,2}(\theta),\label{eq:liketerms}
\end{align}
and the MLE $\hat{\theta}_{n}$ is a maximiser of $L_{n}(\theta)=\sum_{t=1}^{n}\ell_{t}(\theta)$.

\subsection{Asymptotics of the local loglikelihood}

Recalling that $\theta_{0}\in\intr\Theta$ denotes the true values
of the model parameters, we may define the centred loglikelihood with
respect to the \emph{local} parameters $\loc\defeq n^{1/2}(\theta-\theta_{0})\in\reals^{d_{\theta}}$
as
\[
\mathcal{L}_{n}(\loc)=L_{n}(\theta_{0}+n^{-1/2}\loc)-L_{n}(\theta_{0})=\mathcal{L}_{n,1}(\loc)+\mathcal{L}_{n,2}(\loc)
\]
where
\[
\mathcal{L}_{n,i}(\loc)\defeq\sum_{t=1}^{n}[\ell_{t,i}(\theta_{0}+n^{-1/2}\loc)-\ell_{t,i}(\theta_{0})]
\]
for $i\in\{1,2\}$. Observe that $\mathcal{L}_{n,2}(\loc)$ is quadratic.
Moreover, by a Taylor expansion, there exists a $C<\infty$ (depending
on $\theta_{i,0}$) such that for all $n$ sufficiently large
\[
\abs{\log(\theta_{i,0}+n^{-1/2}\loc_{i})-\log(\theta_{i,0})-\frac{n^{-1/2}\loc_{i}}{\theta_{i,0}}+\frac{n^{-1}\loc_{i}^{2}}{2\theta_{i,0}^{2}}}\leq n^{-3/2}C\smlabs{\loc_{i}}^{3}.
\]
Therefore for each (fixed) $\loc\in\reals^{d_{\theta}}$,
\begin{equation}
\mathcal{L}_{n}(\loc)=\loc^{\trans}\mathcal{S}_{n}+\frac{1}{2}\loc^{\trans}\mathcal{H}_{n}\loc+O_{p}(n^{-1/2})\eqdef\mathcal{M}_{n}(\loc)+O_{p}(n^{-1/2})\label{eq:likeapprox}
\end{equation}
where
\begin{align*}
\mathcal{S}_{n} & \defeq\frac{1}{n^{1/2}}\sum_{t=1}^{n}s_{t}\defeq\frac{1}{n^{1/2}}\sum_{t=1}^{n}\grad_{\theta}\ell_{t}(\theta_{0}), & \mathcal{H}_{n} & \defeq\frac{1}{n}\sum_{t=1}^{n}h_{t}\defeq\frac{1}{n}\sum_{t=1}^{n}\grad_{\theta}^{2}\ell_{t}(\theta_{0}),
\end{align*}
for $\grad_{\theta}\ell_{t}$ the gradient and $\grad_{\theta}^{2}\ell_{t}$
the Hessian of $\ell_{t}$.

Recall the decomposition of $\log f_{\theta}(z_{t}\mid\b z_{t-1})$
in \eqref{liketerms} above, and let ${\cal N}\subset\intr\Theta$
denote a compact neighbourhood of $\theta_{0}$. Both $\ell_{t,1}$
and $\ell_{t,2}$ are twice continuously differentiable, and there
exists a $C<\infty$ such that
\begin{align}
\sup_{\theta\in\mathcal{N}}\smlnorm{\grad_{\theta}^{i}\ell_{t,1}(\theta)} & <C & \expect\sup_{\theta\in\mathcal{N}}\smlnorm{\grad_{\theta}^{i}\ell_{t,2}(\theta)} & <C\expect\smlnorm{\xi_{t}}^{2}<\infty,\label{eq:derivbnd}
\end{align}
for $i\in\{1,2\}$, with the final inequality holding because $\expect\smlnorm{\b z_{t}}^{m_{0}}$,
and therefore $\expect\smlnorm{\xi_{t}}^{m_{0}}$, is bounded for
all $t$ by \corref{ergodicity}, where $m_{0}>4$. Thus we may exchange
(twice) differentiation of $\ell_{t}$ with the taking of (conditional)
expectations. It follows in particular (cf.\ \citealp{HH80}, p.\ 157)
that 
\[
\expect_{t-1}s_{t}=\left.\expect_{t-1}\grad_{\theta}\log f_{\theta}(z_{t}\mid\b z_{t-1})\right|_{\theta=\theta_{0}}=\grad_{\theta}\left.\int_{\reals^{p}}f_{\theta}(\bar{z}_{t}\mid\b z_{t-1})\diff\bar{z}_{t}\right|_{\theta=\theta_{0}}=0
\]
so that $\{s_{t}\}$ forms a martingale difference sequence, and similarly
\begin{equation}
\expect_{t-1}h_{t}=\expect_{t-1}\grad_{\theta}^{2}\left.\log f_{\theta}(z_{t}\mid\b z_{t-1})\right|_{\theta=\theta_{0}}=-\expect_{t-1}s_{t}s_{t}^{\trans}.\label{eq:infeq}
\end{equation}
By \corref{ergodicity} (and \remref{ergod}\ref{subrem:order}),
the process $(z_{t}^{\trans},\b z_{t-1}^{\trans})^{\trans}$ is a
positive Harris recurrent Markov chain. It therefore follows by Theorem~17.1.7
in \citet{MT09}, and \eqref{infeq}, that
\begin{equation}
\mathcal{H}_{n}=\frac{1}{n}\sum_{t=1}^{n}h_{t}\inprob\mathbf{E}h_{1}=-\mathbf{E}s_{1}s_{1}^{\trans}=-\mathcal{I}(\theta_{0}),\label{eq:hesslim}
\end{equation}
where $\mathbf{E}$ denotes an expectation with respect to the stationary
distribution of $(z_{t}^{\trans},\b z_{t-1}^{\trans})^{\trans}$.

To verify that $\mathcal{I}(\theta_{0})$ is positive definite, observe
that
\[
h_{t}=\grad_{\theta}^{2}\ell_{t,1}(\theta_{0})+\grad_{\theta}^{2}\ell_{t,2}(\theta_{0})=\grad_{\theta}^{2}\ell_{t,1}(\theta_{0})-M^{\trans}(\xi_{t}\xi_{t}^{\trans}\otimes I_{p})M,
\]
where the first r.h.s.\ term is negative semidefinite. Thus $\mathbf{E}h_{t}$
will be negative definite if $\mathbf{E}\xi_{t}\xi_{t}^{\trans}$
is invertible. Suppose for a contradiction that it were not; then
in view of the definition of $\xi_{t}$, there must exist constants
$a_{0}$, $\{a_{y,i}^{+},a_{y,i}^{-},a_{x,i}\}_{i=0}^{k}$, which
are not all zero, such that
\[
a_{0}+\sum_{i=0}^{k}a_{x,i}^{\trans}x_{t-i}+\sum_{i=0}^{k}(a_{y,i}^{+}y_{t-i}^{+}+a_{y,i}^{-}y_{t-i}^{-})=0
\]
holds almost surely under the stationary distribution $\mathbf{P}$
of $(z_{t}^{\trans},\b z_{t-1}^{\trans})^{\trans}$. Since the l.h.s.\ is
continuous in $(z_{t}^{\trans},\b z_{t-1}^{\trans})^{\trans}$, and
by \corref{ergodicity} the stationary distribution of $(z_{t}^{\trans},\b z_{t-1}^{\trans})^{\trans}$
is equivalent to Lebesgue measure on $\reals^{(k+1)p}$, it follows
that the preceding must hold for all $(z_{t}^{\trans},\b z_{t-1}^{\trans})^{\trans}\in\reals^{(k+1)p}$.
Deduce each of $a_{0}$ and $\{a_{y,i}^{+},a_{y,i}^{-},a_{x,i}\}_{i=0}^{k}$
must be zero, a contradiction. Thus $\mathbf{E}\xi_{t}\xi_{t}^{\trans}$
is invertible.

To determine the limiting behaviour of $\mathcal{S}_{n}$, we verify
the conditions of the martingale CLT given as Theorem~3.2 in \citet{HH80}.
We have in view of \eqref{derivbnd} that
\[
n^{-1/2}\max_{t\leq n}\smlnorm{s_{t}}\leq Cn^{-1/2}\left(1+\max_{t\leq n}\smlnorm{\xi_{t}}^{2}\right)=O_{p}(n^{-1/2+2/m_{0}})=o_{p}(1),
\]
since $\sup_{t\in\naturals}\expect\smlnorm{\xi_{t}}^{m_{0}}<\infty$
by \corref{ergodicity}, where $m_{0}>4$, and similarly
\[
\expect\left(n^{-1}\max_{t\leq n}\smlnorm{s_{t}}^{2}\right)\leq n^{-1}C\left(1+\expect\max_{t\leq n}\smlnorm{\xi_{t}}^{4}\right)\leq C_{1}\left(n^{-1}+\max_{t\leq n}\expect\smlnorm{\xi_{t}}^{4}\right)
\]
is bounded uniformly in $n$. Since by Theorem~17.1.7 in \citet{MT09},
\[
\frac{1}{n}\sum_{t=1}^{n}s_{t}s_{t}^{\trans}\inprob\mathbf{E}s_{1}s_{1}^{\trans}=\mathcal{I}(\theta_{0})
\]
it follows by Theorem~3.2 in \citet{HH80} that
\begin{equation}
\mathcal{S}_{n}\wkc\mathcal{S}\sim N[0\sep\mathcal{I}(\theta_{0})].\label{eq:scorelim}
\end{equation}

It follows from \eqref{likeapprox}, \eqref{hesslim} and \eqref{scorelim}
that
\[
\mathcal{L}_{n}(\loc)=\mathcal{M}_{n}(\loc)+o_{p}(1)\wkc\loc^{\trans}\mathcal{S}-\frac{1}{2}\loc^{\trans}\mathcal{I}(\theta_{0})\loc\eqdef\mathcal{M}(\loc)
\]
holds in the sense of finite-dimensional convergence. Since $\mathcal{I}(\theta_{0})$
is positive definite, the strictly concave function $\mathcal{M}$
is uniquely maximised at $\loc^{\ast}=\mathcal{I}(\theta_{0})^{-1}\mathcal{S}$.
Since $\mathcal{L}_{n}$ is itself concave for all $n\in\naturals$,
it follows by Lemma~A in \citet{Kni89CJS} that, for $\hat{\loc}_{n}$
a maximiser of $\mathcal{L}_{n}$,
\[
n^{1/2}(\hat{\theta}_{n}-\theta_{0})=\hat{\loc}_{n}\wkc\loc^{\ast}=\mathcal{I}(\theta_{0})^{-1}\mathcal{S}\sim N[0\sep\mathcal{I}(\theta_{0})^{-1}].\tag*{\qedsymbol}
\]

\section{Computational details}

\label{app:computation}

We assume, as appropriate for the CKSVAR, that the sets $\{\mathscr{W}_{i}\}_{i\in{\cal I}}$
are convex cones that partition $\reals^{d_{w}}$, for which there
exist matrices $\{E_{i}\}_{i\in{\cal I}}$ such that $\abv{\mathscr{W}}_{i}=\{w\in\reals^{d_{w}}\mid E_{i}w\geq0\}$,
where $\abv{\mathscr{W}}_{i}$ denotes the closure of $\mathscr{W}_{i}$.
Let $\mathscr{W}_{ij}\defeq\{w\in\mathscr{W}_{i}\mid A[i]w\in\mathscr{W}_{j}\}$
denote the values in $\mathscr{W}_{i}$ that would be mapped to $\mathscr{W}_{j}$,
and
\[
E_{ij}\defeq\begin{bmatrix}E_{i}\\
E_{j}A[i]
\end{bmatrix},
\]
so that $\abv{\mathscr{W}}_{ij}=\{w\in\reals^{d_{w}}\mid E_{ij}w\geq0\}$.
For a square matrix $A\in\reals^{m\times m}$, let $A\succ0$ ($A\succeq0$)
denote that $A$ is positive definite (semidefinite); $A\geq0$ denotes
that $A$ has only non-negative entries.

\paragraph*{(i). }

Consider taking ${\cal C}$ to consist of functions of the form $\smldblangle w\defeq[w^{\trans}Pw]_{+}^{1/2}$,
where $P\in\reals^{d_{w}\times d_{w}}$, so that $\smldblangle w_{i}=[w^{\trans}P_{i}w]_{+}^{1/2}$
for some $P_{i}\in\reals^{d_{w}\times d_{w}}$, for each $i\in{\cal I}$.
The first requirement of \eqref{RSR} is satisfied if $P_{i}$ is
such that $w^{\trans}P_{i}w>0$ for $w\in\abv{\mathscr{W}}_{i}\backslash\{0\}$
(as opposed to the whole of $\reals^{d_{w}}$); while for $(i,j)\in{\cal J}$
and $w\in\mathscr{W}_{i}$
\begin{align*}
\smldblangle{A[i]w}_{j}\leq\gamma\smldblangle w_{i} & \iff[w^{\trans}A[i]^{\trans}P_{j}A[i]w]_{+}^{1/2}\leq\gamma[w^{\trans}P_{i}w]_{+}^{1/2}\\
 & \iff w^{\trans}A[i]^{\trans}P_{j}A[i]w\leq\gamma^{2}w^{\trans}P_{i}w
\end{align*}
where the second equivalence holds so long as $w^{\trans}P_{i}w\geq0$
on $w\in\mathscr{W}_{i}$, as per the first requirement. Thus for
this choice of $\smldblangle w_{i}$, the requirements of \eqref{RSR}
may be rewritten equivalently as $\{P_{i}\}_{i\in{\cal I}}\subset\reals^{d_{w}\times d_{w}}$
and $\gamma\in\reals_{+}$ being such that
\begin{alignat}{2}
w^{\top}P_{i}w & >0, &  & \forall w\in\abv{\mathscr{W}}_{i}\backslash\{0\}\sep\forall i\in{\cal I},\label{eq:lyapunov}\\
\gamma^{2}w^{\top}P_{i}w-w^{\trans}A[i]^{\trans}P_{j}A[i]w & \geq0, & \quad & \forall w\in\mathscr{W}_{ij}\sep\forall(i,j)\in{\cal J},\label{eq:switching}
\end{alignat}
If the preceding inequalities are required to hold for all $w\in\reals^{d_{w}}\backslash\{0\}$,
then the problem of finding the minimum $\gamma$ that satisfies \eqref{lyapunov}--\eqref{switching}
yields the semidefinite programming (SDP) estimate of the CJSR (\citealp[p.~245]{PEDJ16Auto}).
This is computationally inexpensive, but at the cost of being too
conservative, in the sense of likely providing too high an estimate
of the stability degree of the system; the challenge is to reduce
this conservativeness while retaining the computational convenience
of an SDP problem.

Here we follow the approach of \citet{FTCMM02Auto}, noting that it
is sufficient for \eqref{lyapunov}--\eqref{switching} that
\begin{alignat}{2}
w^{\top}P_{i}w & >w^{\trans}F_{i}w, &  & \forall w\in\reals^{d_{w}}\backslash\{0\},\forall i\in{\cal I},\label{eq:PF}\\
\gamma^{2}w^{\top}P_{i}w-w^{\trans}A[i]^{\trans}P_{j}A[i]w & \geq w^{\trans}G_{ij}w, & \quad & \forall w\in\reals^{d_{w}},\forall(i,j)\in{\cal J},\label{eq:PG}
\end{alignat}
where $\{F_{i}\}_{i\in{\cal I}}$ and $\{G_{ij}\}_{(i,j)\in{\cal J}}$
are such that
\begin{align}
w^{\top}F_{i}w & \geq0\sep\forall w\in\mathscr{W}_{i},\forall i\in{\cal I}; & w^{\top}G_{ij}w & \geq0\sep\forall w\in\mathscr{W}_{ij},\forall(i,j)\in{\cal J}\label{eq:FGconditions}
\end{align}
Taking $F_{i}=G_{ij}=0$ reproduces the strengthened form of \eqref{lyapunov}--\eqref{switching}
used to estimate the CJSR, but the fact that \eqref{FGconditions}
need only hold on strict subsets of $\reals^{d_{w}}$ allows these
to be chosen such that \eqref{PF}--\eqref{PG} may admit solutions
for smaller values of $\gamma$. To ensure that \eqref{FGconditions}
is satisfied, we parametrise $F_{i}$ and $G_{ij}$ as
\begin{align*}
F_{i} & \defeq E_{i}^{\trans}U_{i}E_{i} & G_{ij} & :=E_{ij}^{\top}U_{ij}E_{ij}
\end{align*}
where $U_{i}\ge0$ and $U_{ij}\geq0$ are conformable symmetric matrices
with non-negative entries (see \citealp{FTCMM02Auto}, p.\ 2143).
For given values of $\{U_{i}\}_{i\in{\cal I}}$ and $\{U_{ij}\}_{(i,j)\in{\cal J}}$,
and hence of $\{F_{i}\}$ and $\{G_{ij}\}$, the problem reduces to
one of verifying that, for a given $\gamma$,
\begin{alignat}{2}
P_{i}-F_{i} & \succ0, & \quad & \forall i\in\mathcal{I},\label{eq:firstSDPcondition}\\
\gamma^{2}P_{i}-A[i]^{\top}P_{j}A[i]-G_{ij} & \succeq0, &  & \forall(i,j)\in\mathcal{J}.\label{eq:secondSDPcondition}
\end{alignat}
is feasible, i.e.\ that it admits a solution for $\{P_{i}\}_{i\in{\cal I}}$.
To recognise that this can be put in the form of an SDP feasibility
problem, define $Q_{i}\defeq P_{i}-F_{i}$, so that the l.h.s.\ of
the second set of inequalities becomes
\begin{align*}
\gamma^{2}(Q_{i}+F_{i})-A[i]^{\top}(Q_{j}+F_{j})A[i]-G_{ij} & =\gamma^{2}Q_{i}-A[i]^{\top}Q_{j}A[i]+\gamma^{2}F_{i}-A[i]^{\top}F_{j}A[i]-G_{ij}\\
 & \eqdef\gamma^{2}Q_{i}-A[i]^{\top}Q_{j}A[i]-\tilde{G}_{ij}(\gamma)\eqdef R_{ij}
\end{align*}
where $\tilde{G}_{ij}(\gamma)$ depends on $\gamma$, $\{U_{i}\}$
and $\{U_{ij}\}$. Then the feasibility of \eqref{firstSDPcondition}--\eqref{secondSDPcondition}
is equivalent to the existence of $\{Q_{i}\}_{i\in{\cal I}}$ and
$\{R_{ij}\}_{(i,j)\in{\cal J}}$ such that
\begin{alignat*}{2}
Q_{i}\succ0\sep R_{ij} & \succeq0, &  & \forall i\in{\cal I},\forall(i,j)\in{\cal J}\\
\gamma^{2}Q_{i}-A[i]^{\top}Q_{j}A[i]-R_{ij} & =\tilde{G}_{ij}(\gamma), & \quad & \forall(i,j)\in{\cal J},
\end{alignat*}
which indeed has the form of an SDP feasibility problem \citep[p.~2390]{PJ08LAA}.
(For this problem to indeed admit a solution, the entries of $\{U_{i}\}$
and $\{U_{ij}\}$ need to be such that $\tilde{G}_{ij}(\gamma)$ is
symmetric, as follows if all of the matrices $\{U_{i}\}$ and $\{U_{ij}\}$
are themselves symmetric.) The implied Lyapunov function is
\[
V(w)\defeq\sum_{i\in{\cal I}}\smldblangle w_{i}\indic\{w\in\mathscr{W}_{i}\}=\sum_{i\in{\cal I}}[w^{\trans}P_{i}w]_{+}^{1/2}\indic\{w\in\mathscr{W}_{i}\}=\sum_{i\in{\cal I}}(w^{\trans}P_{i}w)^{1/2}\indic\{w\in\mathscr{W}_{i}\},
\]
where the final equality holds since $w^{\trans}P_{i}w\geq0$ on $\mathscr{W}_{i}$.

\paragraph*{(ii).}

The preceding can be generalised by allowing ${\cal C}$ to consist
of functions of the form $\smldblangle w\defeq[p(w)]_{+}^{1/2m}$,
where $p$ is a sum of squares polynomial of degree $2m$, and so
can be written as $p(w)=w^{[m]\trans}Pw^{[m]}$, where $P\in\reals^{l\times l}$
for $l\defeq{d_{w}+m-1 \choose m}$ and $w^{[m]}\in\reals^{l}$ denotes
the $m$-lift of $w$ (see \citealp[Sec.~2--3]{PJ08LAA}). Letting
$A[i]^{[m]}$ be the (unique) matrix satisfying $A[i]^{[m]}w^{[m]}=(A[i]w)^{[m]}$,
we have
\[
p(A[i]w)=(A[i]w)^{[m]\trans}P(A[i]w)^{[m]}=w^{[m]\trans}A[i]^{[m]\trans}PA[i]^{[m]}w^{[m]}.
\]
By arguments analogous to those given above, when specialised to this
class of functions, the requirements of \eqref{RSR} may thus be written
as
\begin{alignat*}{2}
w^{[m]\trans}P_{i}w^{[m]} & >0, &  & \forall w\in\mathscr{\abv W}_{i}\backslash\{0\}\sep\forall i\in{\cal I},\\
\gamma^{2m}w^{[m]\trans}P_{i}w^{[m]}-w^{[m]\trans}A[i]^{[m]\trans}P_{j}A[i]^{[m]}w^{[m]} & \geq0, & \quad & \forall w\in\mathscr{W}_{ij}\sep\forall(i,j)\in{\cal J}.
\end{alignat*}
Similarly, if we define now
\begin{align*}
F_{i} & \defeq E_{i}^{[m]\trans}U_{i}E_{i}^{[m]} & G_{ij} & :=E_{ij}^{[m]\top}U_{ij}E_{ij}^{[m]}
\end{align*}
for $\{U_{i}\}$ and $\{U_{ij}\}$ symmetric matrices of appropriate
dimension with non-negative entries, we can reduce the problem to
one of exactly the same form as the SDP feasibility problem \eqref{firstSDPcondition}--\eqref{secondSDPcondition},
but with $A[i]^{[m]}$ taking the place of $A[i]$. Setting $U_{i}=U_{ij}=0$
here gives the problem solved by the SOS approximation to the CJSR
(\citealp[p.~246]{PEDJ16Auto}), on which this estimate of the RJSR
(for this $\mathcal{C}$) therefore provides a lower bound. The implied
Lyapunov function is
\[
V(w)\defeq\sum_{i\in{\cal I}}\smldblangle w_{i}\indic\{w\in\mathscr{W}_{i}\}=\sum_{i\in{\cal I}}(w^{[m]\trans}P_{i}w^{[m]})^{1/2m}\indic\{w\in\mathscr{W}_{i}\}.
\]

\end{document}